\documentclass[reqno,12pt]{amsart}
 
\usepackage{amsmath}
\usepackage{amssymb}
\usepackage{amsfonts}

\usepackage{physics}
\usepackage{alltt}
\usepackage{amsthm}
\usepackage{bm}
\usepackage{mathrsfs}
\usepackage{graphicx}
\usepackage{color}
	\usepackage{tikz,pgfplots}
\usepackage{mathtools}
\usepackage[T1]{fontenc}
\usepackage[utf8]{inputenc}
\usepackage{xcolor}
\usepackage[colorlinks,linkcolor=black,urlcolor=black,citecolor=black,anchorcolor=black]{hyperref}
\usepackage[margin=1.5in]{geometry}

\newcommand{\C}{\mathbb C}
\newcommand{\R}{\mathbb R}
\newcommand{\N}{\mathbb N}
\newcommand{\Z}{\mathbb Z}

\newcommand{\supp}{\operatorname{supp}}
\newcommand{\re}{\operatorname{Re}}
\newcommand{\im}{\operatorname{Im}}
\newcommand{\ve}{\varepsilon}

\newcommand{\spec}{\operatorname{Spec}}

\newcommand{\sgn}{\operatorname{sgn}}

\newtheorem{thm}{Theorem}[section]
\newtheorem{prop}[thm]{Proposition}
\newtheorem{lemma}[thm]{Lemma}
\newtheorem{cor}[thm]{Corollary}

\newtheorem*{thm*}{Theorem}
\newtheorem*{lemma*}{Lemma}
\newtheorem*{prop*}{Proposition}

\theoremstyle{definition}

\newtheorem{remark}[thm]{Remark}
\newtheorem*{example*}{Example}

\newtheorem*{dfn*}{Definition}

\theoremstyle{remark}
\makeatletter
\let\save@mathaccent\mathaccent
\newcommand*\if@single[3]{%
  \setbox0\hbox{${\mathaccent"0362{#1}}^H$}%
  \setbox2\hbox{${\mathaccent"0362{\kern0pt#1}}^H$}%
  \ifdim\ht0=\ht2 #3\else #2\fi
  }
\newcommand*\rel@kern[1]{\kern#1\dimexpr\macc@kerna}
\newcommand*\widebar[1]{\@ifnextchar^{{\wide@bar{#1}{0}}}{\wide@bar{#1}{1}}}
\newcommand*\wide@bar[2]{\if@single{#1}{\wide@bar@{#1}{#2}{1}}{\wide@bar@{#1}{#2}{2}}}
\newcommand*\wide@bar@[3]{%
  \begingroup
  \def\mathaccent##1##2{%
    \let\mathaccent\save@mathaccent
    \if#32 \let\macc@nucleus\first@char \fi
    \setbox\z@\hbox{$\macc@style{\macc@nucleus}_{}$}%
    \setbox\tw@\hbox{$\macc@style{\macc@nucleus}{}_{}$}%
    \dimen@\wd\tw@
    \advance\dimen@-\wd\z@
    \divide\dimen@ 3
    \@tempdima\wd\tw@
    \advance\@tempdima-\scriptspace
    \divide\@tempdima 10
    \advance\dimen@-\@tempdima
    \ifdim\dimen@>\z@ \dimen@0pt\fi
    \rel@kern{0.6}\kern-\dimen@
    \if#31
      \overline{\rel@kern{-0.6}\kern\dimen@\macc@nucleus\rel@kern{0.4}\kern\dimen@}%
      \advance\dimen@0.4\dimexpr\macc@kerna
      \let\final@kern#2%
      \ifdim\dimen@<\z@ \let\final@kern1\fi
      \if\final@kern1 \kern-\dimen@\fi
    \else
      \overline{\rel@kern{-0.6}\kern\dimen@#1}%
    \fi
  }%
  \macc@depth\@ne
  \let\math@bgroup\@empty \let\math@egroup\macc@set@skewchar
  \mathsurround\z@ \frozen@everymath{\mathgroup\macc@group\relax}%
  \macc@set@skewchar\relax
  \let\mathaccentV\macc@nested@a
  \if#31
    \macc@nested@a\relax111{#1}%
  \else
    \def\gobble@till@marker##1\endmarker{}%
    \futurelet\first@char\gobble@till@marker#1\endmarker
    \ifcat\noexpand\first@char A\else
      \def\first@char{}%
    \fi
    \macc@nested@a\relax111{\first@char}%
  \fi
  \endgroup
}
\makeatother

\makeatletter
\def\@setauthors{%
  \begingroup
  \def\thanks{\protect\thanks@warning}%
  \trivlist
  \centering\large \@topsep30\p@\relax
  \advance\@topsep by -\baselineskip
  \item\relax
  \author@andify\authors
  \def\\{\protect\linebreak}%
  \authors%
  \ifx\@empty\contribs
  \else
    ,\penalty-3 \space \@setcontribs
    \@closetoccontribs
  \fi
  \endtrivlist
  \endgroup
}
\def\@settitle{\begin{center}%
  \baselineskip14\p@\relax
    \normalfont\LARGE
  \@title
  \end{center}%
}
\makeatother

\pgfplotsset{compat=newest}

\begin{document}

\author{Simon Becker}
\address[Simon Becker]{ETH Zurich, 
Institute for Mathematical Research, 
Rämistrasse 101, 8092 Zurich, 
Switzerland}
\email{simon.becker@math.ethz.ch}

\author{Setsuro Fujii\'e}
\address[Setsuro Fujiie]{Department of Mathematical Sciences, Ritsumeikan University, Ku\-sa\-tsu, 525-8577, Japan}
\email{fujiie@fc.ritsumei.ac.jp}

\author{Jens Wittsten}
\address[Jens Wittsten]{Department of Engineering, University of Bor{\aa}s, SE-501 90 Bor{\aa}s, Sweden}
\email{jens.wittsten@hb.se}

\title{Bohr--Sommerfeld rules for systems}

\date{\today}

\begin{abstract}
We present a complete, self-contained formulation of the Bohr--Sommerfeld quantization rule for a semiclassical self-adjoint $2 \times 2$ system on the real line, arising from a simple closed curve in phase space. We focus on the case where the principal symbol exhibits eigenvalue crossings within the domain enclosed by the curve -- a situation commonly encountered in Dirac-type operators. Building on earlier work on scalar Bohr--Sommerfeld rules and semiclassical treatments of the Harper operator near rational flux quanta, we derive concise expressions for general self-adjoint $2 \times 2$ systems. The resulting formulas give explicit geometric phase corrections and clarify when these phases take quantized values.
\end{abstract}

\maketitle

\section{Introduction}

We consider a semiclassical self-adjoint $2\times2$ system $H^w(x,hD)$ on the real line, with Weyl symbol $H(x,\xi)\sim \sum_{j=0}^\infty h^j H_j(x,\xi)$. Here the $H_j$'s belong to some appropriate symbol classes, such as $H_j\in S(m)$ for some weight function $m$, or $H_j\in S^{m-j}$ for some $m\in\R$. The principal symbol can be written as
\begin{equation}\label{eq:H0}
H_0  = \sum_{i=0}^3 \sigma_i p_i =\begin{pmatrix} p_0+p_3 & p_1-ip_2\\ p_1+ip_2& p_0-p_3\end{pmatrix}
\end{equation}
for some real-valued $p_i(x,\xi)\in C^\infty(T^*\R)$, where the $\sigma_i$'s are the Pauli matrices. Introduce the vector $P = (p_1,p_2,p_3)$. Then the eigenvalues of $H_0(x,\xi)$ are
\[ \lambda_{\pm}(x,\xi) = p_0(x,\xi)\pm \lVert P(x,\xi) \rVert,\]
and we see that there is an eigenvalue crossing if and only if $P$ vanishes somewhere. 

Let $\mu$ be one of $\lambda_{\pm}$, and $E\in\R$ an energy level. The goal of this paper is to present a complete and self-contained formulation of the Bohr--Sommerfeld rule coming from a simple closed curve $\gamma=\mu^{-1}(E)$ in phase space. We assume $d\mu\ne0$ and $P\ne0$ in a neighborhood of $\gamma$, but allow $P$ to vanish inside the domain $D$ enclosed by $\gamma=\partial D$. This behavior is typical for Dirac-type systems. One motivation for our work is to develop a framework for understanding the occurrence of almost flat bands in the spectrum of a Dirac-Harper model for strained moiré lattices, introduced by Timmel and Mele \cite{TM2020} and discussed in Subsection \ref{ss:TM}. In doing so, we build on earlier work on scalar Bohr--Sommerfeld rules by Helffer-Robert  \cite{helffer1983calcul,helffer1982asymptotique,helffer1984puits} together with Helffer-Sjöstrand's semiclassical treatments of the Harper operator \cite{HS88,HS90} near rational fluxes, which we revisit to derive concise expressions for general self-adjoint $2 \times 2$ systems.

If $P\ne0$ it is natural to diagonalize $H_0$ and reduce to the scalar case. If $P=0$ somewhere inside $D$, we first remove the eigenvalue crossings by slightly perturbing the components of $P$, see Lemma \ref{lem:perturbation}. We also modify $H_0$ away from a simply connected neighborhood of $\gamma$ to ensure $H_0$ is bounded, see Lemma \ref{lem:boundedsymbol}. Both these modifications can be done in a way that doesn't change the spectrum of $H^w(x,hD)$ modulo $\mathcal O(h^\infty)$, see Proposition \ref{prop:samespectrum}.
We can then find $U\in S(1)$ such that $U^w$ is unitary and
$$
(U^w)^*H^w U^w=\begin{pmatrix}\mu^w& \\ & A_{22}^w\end{pmatrix}+hD^w+\mathcal O_{L^2\to L^2}(h^\infty),
$$
where $D=\operatorname{diag}(D_{11},D_{22})$ is diagonal, see the discussion after Theorem \ref{thm:Harper2}. Here, $A_{22}$ is the other eigenvalue of $H_0$, that is, if $\mu=\lambda_+$ then $A_{22}=\lambda_-$ and vice versa. Since $d\mu\ne0$ near $\gamma=\mu^{-1}(E)$ we can use the analysis of Helffer--Robert \cite{helffer1983calcul,helffer1982asymptotique,helffer1984puits} (see also Sjöstrand \cite[Theorem 8.4]{sjostrand2006cime}) for scalar operators of principal type to obtain a Bohr--Sommerfeld rule that describes the spectrum of $H^w$ coming from $\gamma$. This is established for the modified operator in Theorem \ref{thm:main}, and as indicated above it leads to the same Bohr--Sommerfeld rule being valid for the original operator, see Theorem \ref{thm:main0}. This rule takes the form
\begin{equation}\label{eq:BS}
2\pi kh=S(E)\sim\sum_{j=0}^\infty S_j(E)h^j,
\end{equation}
where the right-hand side is the semiclassical action, consisting of the action integral along $\gamma$, that is,
\begin{equation}\label{eq:S0}
S_0(E)=\int_\gamma \xi\,dx,
\end{equation}
and where
\begin{equation}\label{eq:S1}
S_1(E)=\pi-\int_\gamma f_1\,dt
\end{equation}
contains the Maslov index in the $\pi$ term. Here $f_1$ is the subprincipal symbol of $\mu^w+hD_{11}^w$, and the integral $\int_\gamma f_1\,dt$ contains the geometric phase corrections in the Bohr--Sommerfeld rule. Each subsequent $S_j(E)$, $j\ge2$ can also be computed using an algorithm due to Colin de Verdière \cite{cdv}.

To illustrate the role of subprincipal terms, consider the Dirac system
\begin{equation}\label{eq:simpleDirac}
\begin{pmatrix} 0 & hD_x-ix \\ hD_x+ix& 0\end{pmatrix}
\end{equation}
whose off-diagonal entries are the annihilation and creation operators of the harmonic oscillator.
This simple model can be seen as a special case of the Jackiw-Rebbi model studied in Subsection \ref{sec:JR}. It is known that the operator eigenvalues of \eqref{eq:simpleDirac} coincide with the roots to $e^{iS_0(E)/h}=1$, where $S_0$ is given by \eqref{eq:S0}. (The corresponding eigenfunctions can be computed explicitly.) The eigenvalues of the principal symbol are $\lambda_\pm(x,\xi)=\pm\sqrt{\xi^2+x^2}$, and the level sets $\lambda_\pm=E$ are circles with radius $\lvert E\rvert$. Hence, $S_0(E)=\pi E^2$ which by \eqref{eq:BS} gives $$E=\pm\sqrt{2kh},\quad k\in\N.$$  
However, if we instead diagonalize the symbol near $\lambda_\pm^{-1}(E)$ and apply the scalar Bohr--Sommerfeld rule to the diagonal elements of the principal part $\operatorname{diag}(\lambda_+^w, \lambda_-^w)$, we obtain
$$
E=\pm\sqrt{(2k+1)h},\quad k\in\N,
$$
which has an incorrect offset. Hence, there has to be a contribution encoded in the lower-order corrections of the diagonalized operator that removes this discrepancy. The missing contribution comes from the subprincipal part $f_1$, which consequently has to be calculated precisely.

This kind of analysis has been used by Helffer and Sjöstrand in their treatment of the Harper operator (see \cite[Section 3.6]{HS90} in particular for operators such as \eqref{eq:simpleDirac}).
As in \cite{HS88,HS90} we note that the analysis presented here generalizes e.g.~to the case when $H_0$ is periodic and to many other situations where $\mu^{-1}(E_0)$ is not just one simple closed curve but instead a countable union of connected components, provided these components are separated by barriers, see Remark \ref{rem:manyconnectedcomponents}.

Before stating our main results, which include concise expressions for the geometric phase corrections of the Bohr--Sommerfeld rule, together with information about when these phases only take a discrete set of values, i.e., when they become {\it quantized}, we briefly note the recent work of Yoshida  \cite{Yoshida}, where a different and inherently non-scalar method was used to obtain the leading-order term in the Bohr--Sommerfeld rule for operators with $p_2=\xi$, $p_3=V(x)$, and $p_0=p_1=0.$ The new ingredient in Yoshida's method is that it applies when $V(x)\in C^\infty$ is not necessarily analytic. When $V(x)$ is analytic, the exact WKB method had previously been successfully used to obtain Bohr--Sommerfeld rules for various $2\times 2$ systems, including some non-self-adjoint cases \cite{fujiie2020semiclassical,fujiie2009semiclassical,fujiie2018quantization,Hirota2017Real,hirota2021complex}. We also mention that the Bohr--Sommerfeld rules we present below are for energies away from the eigenvalue crossing point -- for results on spectrum for energies near these critical points, where the eigenvalues of the principal symbol coalesce, we refer for example to \cite{becker2022semiclassical,cornean2021spectral,HS90} and the references therein. It follows from \cite{HS88} (see also \cite{hz} for a recent treatment of a more general case) that, for scalar operators with a non-degenerate principal symbol near the bottom of the well, the Bohr–Sommerfeld quantization rule remains valid at the bottom of the spectrum. In the case of systems, this applies to operators such as \eqref{eq:simpleDirac}, at the eigenvalue crossing, and to their generalizations studied in \S 3.6 of \cite{HS90}, but it may fail for general $2\times 2$ systems that we study in this work.

\subsection{Statement of results}
We make the following assumptions.
Let $\mu$ be one of $\lambda_{\pm}$. Fix $E_0\in\R$ and assume that $\mu^{-1}(E_0)=\gamma_0$ for a simple closed curve $\gamma_0$ such that $d\mu\ne0$ and $P\ne0$ near $\gamma_0$.
\begin{itemize}
\item We fix an interval $I = [E_-,E_+] \subset \mathbb R$, with $E_- < E_0< E_+$,  and we assume there is a topological ring $\mathcal A$ such that 
$\partial \mathcal A = A_- \cup A_+$ with $A_{-},A_+$ the connected components of $\mu^{-1}(E_{\pm}).$
\item We assume $\mu$ has no critical points in $\mathcal A.$
\item We assume without loss of generality that $A_-$ is included in the disk $W$ enclosed by $A_+$ ($W$ is called the {\it well}), otherwise we can study $-H^w$ instead. 
\item We assume (possibly after shrinking $I$) that if $\mu=\lambda_+$ then 
\begin{equation}\label{eq:gaptypeassumptionplus}
\lambda_-(x,\xi)<E_-\quad\text{ for all }(x,\xi)\in \mathcal A.
\end{equation}
If instead $\mu=\lambda_-$ we assume that
\begin{equation}\label{eq:gaptypeassumptionminus}
\lambda_+(x,\xi)>E_+\quad \text{ for all }(x,\xi)\in \mathcal A.
\end{equation}
%\begin{equation}\label{eq:gaptypeassumptionplus}
%\mu(x,\xi)\in I\text{ for }(x,\xi)\in \mathcal A\quad\Longrightarrow\quad \lambda_-(x,\xi)<E_-.
%\end{equation}
%If instead $\mu=\lambda_-$ we assume that
%\begin{equation}\label{eq:gaptypeassumptionminus}
%\mu(x,\xi)\in I\text{ for }(x,\xi)\in \mathcal A\quad\Longrightarrow\quad \lambda_+(x,\xi)>E_+.
%\end{equation}
\end{itemize}
Two such situations are illustrated in Figure \ref{fig:well}.
\begin{figure}
\begin{center}
\includegraphics[width=0.9\textwidth]{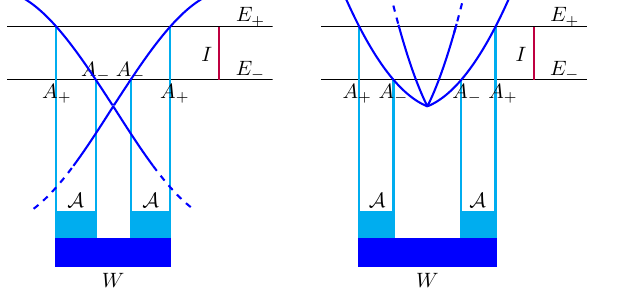}
\end{center}
\caption{\label{fig:well} Two cases that illustrate the definitions of the interval $I=[E_-,E_+]$ together with the well $W$ and the topological ring $\mathcal A$ such that $\partial A=A_+\cup A_-$ with $A_\pm$ the connected components of $\mu^{-1}(E_\pm)$. On the left $\mu=\lambda_+=p_0+\lVert P\rVert $ and $\lVert P\rVert $ dominates the behavior, while on the right $\mu=\lambda_-=p_0-\lVert P\rVert $ and the behavior is dominated by $p_0$. Condition \eqref{eq:gaptypeassumptionplus} is clearly satisfied in the left panel, and condition \eqref{eq:gaptypeassumptionminus} in the right. }
\end{figure}
Under these assumptions, we define the curve 
\begin{equation}\label{eq:curve}
\gamma=\gamma(E):=\mu^{-1}(E)\cap\mathcal A.
\end{equation}
Thus $\gamma$ is a curve close to $\gamma_0$ for $E\in I$.
We orient $\gamma$ along the positive flow direction of the Hamilton vector field $H_\mu=\partial_\xi\mu\partial_x-\partial_x\mu\partial_\xi$. When $\mu$ has a well, this means that $\gamma$ is oriented in the clockwise direction. If $T$ is a minimal period of $\gamma$ and $[0,T)\ni t\mapsto (x(t),\xi(t))$ a parametrization, then for any function $f$ defined near $\gamma$ we write $\int_\gamma f\,dt$ to denote $\int_0^T f(x(t),\xi(t))\,dt$. We let $\langle\phantom{i},\phantom{i}\rangle$ be the inner product in $\C^2$, and write $\{ A, B\}=\partial_\xi A\cdot\partial_x B-\partial_x A\cdot\partial_\xi B$ for any matrices $A$ and $B$ for which the product is well-defined.

\begin{thm}\label{thm:main0}
Under the assumptions above, $\spec(H^w(x,hD))\cap[E_-,E_+]$ is modulo $\mathcal O(h^\infty)$ described by the Bohr--Sommerfeld rule \eqref{eq:BS} arising from $\gamma$, i.e.,
\begin{equation}\label{eq:main0}
2\pi kh=S(E)+\mathcal O(h^\infty)\sim\sum_{j=0}^\infty S_j(E)h^j,%+\mathcal O(h^\infty),
\end{equation}
where $S_0(E)=\int_\gamma\xi\,dx$ and $S_1(E)=\pi-\int_\gamma f_1\,dt$ as in \eqref{eq:S0} and \eqref{eq:S1}. Let $e$ be a smooth normalized eigenvector corresponding to $\mu=\lambda_\pm$ defined near $\gamma$, and write $H_1=\sum_{i=0}^3 r_i\sigma_i$ for real-valued $r_i\in C^\infty(T^*\R)$. Then
\begin{equation*}
 f_1= r_0\pm\sum_{i=1}^3\frac{ r_ip_i}{\lVert P\rVert}+
\mu_1
\end{equation*}
where 
\begin{equation}\label{eq:subprincipalterm}
\mu_1=\frac{1}{2i}\big\langle \{H_0-\mu,e\},e\big\rangle+\frac{1}{i}\big\langle \{\mu,e\},e\big\rangle .
\end{equation}
\end{thm}

We prove this theorem in Section \ref{sec:preparation} by microlocally preparing the original system and then reducing to the scalar case, as explained above. We note that the results of Section \ref{sec:preparation} (and Theorem \ref{thm:main0} in particular) can easily be generalized to $n\times n$ systems with only minor modifications to the assumptions and arguments, and our presentation in that section has been written with this in mind. (For example, if $\mu$ were an eigenvalue of an $n\times n$ matrix-valued symbol, one would introduce some combination of assumptions \eqref{eq:gaptypeassumptionplus} and \eqref{eq:gaptypeassumptionminus}.) In Section \ref{sec:subprincipal} we instead take advantage of the $2\times 2$ structure to analyze the subprincipal contribution in \eqref{eq:subprincipalterm} in greater detail.

In \cite{HS90}, the two terms on the right of \eqref{eq:subprincipalterm} are denoted by $\mu_1'$ and $\mu_1''$ respectively,  see formula $(6.2.19)$ in \cite{HS90}. Accordingly, we denote them by $\mu_1'$ and $\mu_1''$ as well.
Then $\mu_1''$ is the Berry connection, which upon integration gives rise to the {\it Berry phase}
$$
\theta_B=\int_\gamma\mu_1''\,dt,
$$
see Remark \ref{rem:Berrysphase}. The term $\mu_1'$ is related to the Berry curvature scalar density, and upon integration gives rise to the {\it Rammal--Wilkinson phase}
$$
\theta_{\mathit{RW}}=\int_\gamma \mu_1'\,dt,
$$
see Remark \ref{rem:RWphase}. Together, $\theta_B$ and $\theta_{\mathit{RW}}$ provide the geometric phase corrections to the Bohr--Sommerfeld rule.

To describe $\theta_B$ and $\theta_{\mathit{RW}}$ more explicitly, we will introduce spherical coordinates to represent the vector $P=(p_1,p_2,p_3)$.
To avoid trivial cases we will assume without loss of generality that neither $p_1$ nor $p_2$ vanishes identically near $\gamma$. (If any two components of $P$ vanish identically then $H_0$ could be reduced to a diagonal matrix, and if one component vanishes identically we can always rotate $P$ so that it's the third component that vanishes, see the discussion preceding Theorem \ref{thm:quantizedmu1} in Section 4.) We can then write $P = \lVert P \rVert (\sin\theta \cos\phi, \sin\theta \sin\phi, \cos\theta)$ with
\begin{equation}\label{eq:sphericalcoordinates}
\left\{
\begin{aligned}
\theta &= \arccos\left( p_3/\lVert P \rVert \right), \\
\phi &= \sgn(p_2) \arccos\left( p_1/(p_1^2 + p_2^2)^{1/2} \right).
\end{aligned}
\right.
\end{equation}
(At points where $p_2(x,\xi)=0$ this is interpreted in the standard way as $\phi=0$ if $p_1(x,\xi)>0$, and $\phi=\pi $ if $p_1(x,\xi)<0$.)
The principal symbol then takes the form
\begin{equation}\label{eq:easyH}
H_0 = p_0 I + \lVert P \rVert R, \quad 
R := \begin{pmatrix} \cos\theta & e^{-i\phi}\sin\theta \\ e^{i\phi}\sin\theta & -\cos\theta \end{pmatrix},
\end{equation}
where $R$ is unitary and Hermitian.

\begin{thm}\label{thm:main1}
Assume that
neither $p_1$ nor $p_2$ in \eqref{eq:H0} vanishes identically near $\gamma$,
and let $\theta$ and $\phi$ be spherical coordinates given by \eqref{eq:sphericalcoordinates}. If $\mu=\lambda_\pm$ then
\begin{align*}
\theta_B&=\pm\int_\gamma  \frac{1-\cos(\theta)}{2}\{\lambda_\pm,\phi\}\,dt,\quad \theta_{\mathit{RW}}=\int_\gamma \lVert P\rVert\frac{\sin(\theta)}{2}\{\theta,\phi\}\,dt,
\end{align*}
where the integrals are independent of the choice of eigenvectors modulo $2\pi\Z$.
\end{thm}

As we can see from Theorem \ref{thm:main0}, unless there are restrictions on the subprincipal symbol $H_1$, the term $\int_\gamma f_1 \, dt$ will not be quantized in general. 
However, in Section \ref{sec:phasecalculations} we show that the Berry and Rammal--Wilkinson phases are quantized as soon as the components of $P$ are linearly dependent over $\R$, see Theorem \ref{thm:quantizedmu1}.  In particular, this happens when (at least) one of $p_1,p_2,p_3$ vanishes identically near $\gamma$. 
The following theorem therefore complements Theorem \ref{thm:main1}, and together they yield a comprehensive description of the general situation.

To state the result we let $\operatorname{wind}(\Gamma,0)$ denote the winding number of a curve $\Gamma\subset \C$ around the origin in the complex plane. Given a complex-valued function $q:T^*\R\to \C$ we let $q(\gamma)$ be the image under $q$ of $\gamma$ with the induced orientation. We then have the following special case of Theorem \ref{thm:quantizedmu1}:

\begin{thm}\label{thm:main2}
Assume that $p_i\equiv0$ near $\gamma$ for some $i\in\{1,2,3\}$. If $\mu=\lambda_\pm$ then
\begin{align*}
\theta_B=\pm\pi \operatorname{wind}(\Gamma_i,0),\quad \theta_{\mathit{RW}}=0,
\end{align*}
where $\Gamma_1=(-p_3+ip_2)(\gamma)$, $\Gamma_2=(p_1-ip_3)(\gamma)$, and $\Gamma_3=(p_1+ip_2)(\gamma)$.
\end{thm}

Note that if for example $p_1\equiv p_2\equiv0$ then we get $\theta_B=0$ using either $\Gamma_1$ or $\Gamma_2$ to compute the winding number, so there is no ambiguity.
Combining Theorem \ref{thm:main0} with Theorems \ref{thm:main1} and \ref{thm:main2} we immediately obtain the following corollary.

\begin{cor}\label{cor:main0}
Let $H\sim H_0+hH_1+\ldots$ be as above, with $H_1=\sum_{i=0}^3 r_i\sigma_i$ for real-valued $r_i\in C^\infty(T^*\R)$.
If  $\mu=\lambda_\pm$ then $\spec(H^w(x,hD))\cap[E_-,E_+]$ is described by the Bohr--Sommerfeld rule
\begin{equation*}
2\pi kh=\sum_{j=0}^\infty S_j(E)h^j+\mathcal O(h^\infty),
\end{equation*}
where $S_0(E)=\int_\gamma\xi\,dx$ and 
\begin{equation*}
S_1(E)=\pi-\int_\gamma\bigg(r_0\pm\sum_{i=i}^3\frac{ r_ip_i}{\lVert P\rVert}\bigg)\,dt-\theta_B-\theta_{\mathit{RW}},
\end{equation*}
with $\theta_B,\theta_{\mathit{RW}}$ as in Theorem \ref{thm:main1} if neither $p_1$ nor $p_2$ vanishes identically near $\gamma$, and with $\theta_B,\theta_{\mathit{RW}}$ as in Theorem \ref{thm:main2} otherwise.
\end{cor}

We also record the Bohr--Sommerfeld rule arising from $\gamma$ when $\mu$ has a barrier instead of a well in the domain enclosed by $\gamma$. As alluded to in the beginning of this subsection, it is obtained from the one describing the spectrum of $-H^w$. 

\begin{thm}[Barrier vs.~well]\label{thm:barrier}
Suppose all assumptions above are in force, except assume now that $\mu$ exhibits a barrier instead of a well in the domain $D$ enclosed by $\gamma$. Orient $\gamma$ along the flow of the Hamilton vector field of $\mu$ so $\gamma$ is oriented counterclockwise. If  $\mu=\lambda_\pm$ then $\spec(H^w(x,hD))\cap[E_-,E_+]$ is described by the Bohr--Sommerfeld rule 
$$
2\pi kh=-S_0(E)+\sum_{j=1}^\infty S_j(E)h^j+\mathcal O(h^\infty),
$$
where each $S_j$ is defined by the same formal expressions as in the case when $\mu$ has a well, with the understanding that $\gamma$ is oriented along the Hamilton flow. 
\end{thm}

\begin{proof}
The principal symbol of $-H^w$ is $-H_0$ with eigenvalues $$\mu_\pm:=-\lambda_\mp=-p_0\pm\lVert P\rVert.$$
If $\mu= \lambda_\pm$ then $\mu_\mp$ has a well near the energy level $-E$, and $-E\in\spec(-H^w)\cap[-E_+,-E_-]$ is described by the Bohr--Sommerfeld rule in Theorem \ref{thm:main0} for $\mu_\mp$, with $E$ replaced by $-E$, $\gamma$ by $-\gamma$ (that is, $-\gamma$ is the same curve but oriented in the clockwise direction), and $f_1$ by $-f_1$. Now 
$$
S_0^{\mathrm{well}}(-E)=\int_{\{\mu_\mp=-E\};\text{ clockwise}}\xi\,dx=\int_{-\gamma}\xi\,dx=-\int_\gamma\xi\,dx,
$$
and
$$
S_1^{\mathrm{well}}(-E)=\pi-\int_{-\gamma}(-f_1)\,dt=\pi-\int_\gamma f_1\,dt.
$$
Hence, if we define $S_0^{\mathrm{barrier}}(E)$ and $S_1^{\mathrm{barrier}}(E)$ by the same formal expressions as in the well case, so that $S_0^{\mathrm{barrier}}(E)=\int_\gamma\xi\,dx$ and $S_1^{\mathrm{barrier}}(E)=\pi-\int_\gamma f_1\,dt$, then 
$$
2\pi k h=-S_0^{\mathrm{barrier}}(E)+hS_1^{\mathrm{barrier}}(E)+\sum_{j=2}^\infty S_j^{\mathrm{well}}(-E)h^j.
$$
Using the algorithm of Colin de Verdière \cite{cdv} applied to $-H^w$ it follows from Remark 1 and Theorem 1 in the mentioned paper that $S_j^{\mathrm{well}}(-E)=S_j^{\mathrm{barrier}}(E)$ for $j\ge2$ too if we define $S_j^{\mathrm{barrier}}$ by the same formal expressions as in the case of a well. 
\end{proof}

We round off the paper by illustrating our results in Section \ref{sec:examples}, where we calculate the Bohr--Sommerfeld rules for a few different models, and compare the spectrum predicted by these rules to numerical spectral computations. We include examples both of models where the geometric phase corrections are quantized, and models where they are not.

\section{Microlocal preparation}\label{sec:preparation}
Here we diagonalize the operator $H^w(x,hD)$ by conjugating it with a smooth system. This requires the principal symbol $H_0$ in \eqref{eq:H0} to have uniformly gapped eigenvalues $\lambda_\pm$ everywhere. As noted, there is an eigenvalue crossing if and only if $P$ vanishes somewhere. Thus, we can easily remove eigenvalue crossings in $D$ by slightly perturbing the components of $P$ there.

\begin{lemma}\label{lem:perturbation}
Let $P \in C^{\infty}( \mathbb R^2 ;\mathbb R^3)$ and assume there is a smooth simple closed curve $\gamma$ enclosing an open domain $D$ with $\lVert P(w)\rVert \ge c>0$ for $w \in \gamma$. Then for any $\ve>0$ there is an open set $\Omega\Subset D$ with $\mathcal A\subset \complement \Omega$, and a vector $Q \in C^{\infty}( \mathbb R^2 ;\mathbb R^3)$ with $Q(w) =P(w)$ for $w \notin \Omega$, such that $Q$ is non-vanishing in a neighborhood of $W$, and the modification by $Q$ is small: $\sup_{w \in \{Q(w) \neq P(w)\}} \lVert Q(w) \rVert <\varepsilon c$.
\end{lemma}
\begin{proof}
Let $\mathcal P_{k}:= \{w :\lVert P(w)\rVert<\varepsilon k\}$. Take a cutoff function $\chi \in C_c^{\infty}(D \cap \mathcal P_{3c/4};[0,1])$ such that $\inf_{x \in D \cap \mathcal P_{c/2}} \chi(w) >0.$
We then define the family 
\[ P_v(w):=P(w)+\chi(w) v \text{ with } v\in \mathbb R^3.\]
By choosing $v \in B_0(1/K)$ with $K:= \sup_{w \in \supp(\chi)} \left\lvert \frac{\chi(w)}{\varepsilon c-\lVert P(w)\rVert}\right\rvert$ we have for $w  \in \supp(\chi)$ that
\[  \lVert P_v(w)\rVert \le \lVert P(w)\rVert + \lvert \chi(w)\rvert \lvert v\rvert \le \varepsilon c.\]
Then, $P_v(w)=0$ implies $v = f(w):=-P(w)/\chi(w).$
Since $f: \{w: \chi(w) \neq 0\} \to \mathbb R^3$ it follows that the image of the set of critical values $w$ such that 
\[ f(w)=v \text{ and } Df(w) \text{ is not surjective }\]
has measure zero (Sard's theorem). Since $f$ maps from a domain in $\mathbb R^2$ to $\mathbb R^3$, $Df(w)$ is never surjective, so all $\{w: \chi(w) \neq 0\}$ are critical. Thus, the image of $f$ has measure zero in $\mathbb R^3.$ 
It follows that every $v \in B_0(1/K) \setminus \operatorname{ran}(f)$ leads to an admissible $Q:=P_v.$
\end{proof}

We shall use a vector $Q$ as in the lemma to redefine $H_0$ inside $D$ and show that it doesn't change the spectrum of $H^w(x,hD)$ modulo $\mathcal O(h^\infty)$. But first, we also modify $H_0$ outside the well $W$ so that $H_0-E$ becomes elliptic there with uniformly gapped eigenvalues.

\begin{lemma}\label{lem:boundedsymbol}
There is a neighborhood $\Omega'$ of $W$ and a matrix-valued $\widetilde H_0\in S(1)$ such that $\widetilde H_0-E$ is uniformly elliptic in $\complement\Omega'$ for all $E\in I$ with eigenvalues $\widetilde\lambda_\pm$ that are uniformly gapped in $\complement\Omega'$, while $\widetilde H_0=H_0$ in $\Omega'$. If $\gamma=\lambda_\pm^{-1}(E)=\widetilde \lambda_\pm^{-1}(E)$ then $\widetilde \lambda_\mp^{-1}(E)=\emptyset$ for all $E\in I$.
\end{lemma}

\begin{proof}
We give the proof in the case when $\mu=\lambda_+$, and use \eqref{eq:gaptypeassumptionplus}. The case $\mu=\lambda_-$ is similar, see \eqref{eq:gaptypeassumptionminus}. Since $\mu(A_+)=E_+$ we can take a cutoff function $\psi\in C_c^\infty$ with $\psi=1$ in a neighborhood of $W$ where $\lambda_+>\lambda_-$ (possibly after using Lemma \ref{lem:perturbation} to remove zeros of $P$ inside $D$), 
such that if $w\in\supp\nabla\psi$ then 
$\lambda_+(x,\xi)\ge E_++\ve$ and $\lambda_-(x,\xi)\le E_--\ve$ for some $\ve>0$. Introduce modified spherical coordinates 
$\widetilde\theta = \operatorname{arccos}(\psi p_3/\lVert P\rVert )$ and $\widetilde\phi = \sgn(p_2) \operatorname{arcccos}(\psi p_1/\sqrt{p_1^2+p_2^2})$, and define $\widetilde R$ as $R$ in \eqref{eq:easyH} but with $\theta,\phi$ replaced by $\widetilde\theta,\widetilde\phi$. 
Set
$$
\widetilde H_0=\psi p_0+(1-\psi)\frac{E_++E_-}2+\bigg(\psi\lVert P\rVert+(1-\psi)\bigg(\frac{E_+-E_-}2+\ve\bigg)\bigg)\widetilde R.
$$
Then $\widetilde H_0\in S(1)$ and $\widetilde H_0=H_0$ when $\psi=1$. The eigenvalues of $\widetilde H_0$ are
$$
\widetilde \lambda_\pm
=\psi\lambda_\pm+(1-\psi)(E_\pm\pm\ve),
$$
and it is easy to check that $\widetilde\lambda_\pm$ are uniformly gapped in $\supp(1-\psi)$. We can choose $\psi$ so that if $\psi(w)=1$ then $\widetilde\lambda_-(w)\notin I$. It then follows that $\widetilde\lambda_-^{-1}(E)=\emptyset$ for all $E\in I$.
Now,
$$
\operatorname{det}(\widetilde H_0-E)=\Big(\psi \lambda_++(1-\psi)(E_++\ve)-E\Big)\Big(\psi \lambda_-+(1-\psi)(E_--\ve)-E\Big),
$$
and it is easy to see that if $E\in I$ then this is $\le -\ve(E_+-E_-+\ve)<0$ in $\supp(1-\psi)$, so $\widetilde H_0-E$ is uniformly elliptic in $\supp(1-\psi)$.
\end{proof}

Let $Q=(q_1,q_2,q_3)$ be as in Lemma \ref{lem:perturbation}, and define 
$$
H_0^Q=p_0 \sigma_0+ \sum_{i=1}^3 \sigma_i q_i.
$$
We now modify $H_0^Q$ outside a neighborhood $\Omega'$ of the well $W$ as in Lemma \ref{lem:boundedsymbol}. We lift the construction to $H^w$ by setting 
$$
\widetilde H=\widetilde H_0^Q+\sum_{j=1}^\infty h^j\widetilde H_j,
$$
where $\widetilde H_j=\psi H_j$ with $\psi\in C_c^\infty$ such that $\psi=1$ in $\Omega'$. Then $\widetilde H_j\in S(1)$ for all $j$, and the modified operator $\widetilde H^w$ gives an accurate description of the spectrum of $H^w$ in $[E_-,E_+]$. More precisely, if $d_H(X,Y)$ is the Hausdorff distance between sets $X,Y\subset\C$, then the following holds.

\begin{prop}\label{prop:samespectrum}
Let $X=\spec(H^w)\cap[E_-,E_+]$ and $Y=\spec(\widetilde H^w)\cap[E_-,E_+]$. Then 
$$
d_H(X,Y)=\mathcal O(h^\infty).
$$
\end{prop}

\begin{proof}
This follows from uniform ellipticity of $\widetilde H_0^Q-E$, $E\in I$, away from the topological ring $\mathcal A$. For a detailed proof, we refer to the analysis of Helffer and Sjöstrand \cite[Section 2]{HS88}, see in particular their Proposition 2.7. For an approach that is similar but doesn't use the FBI transform, see \cite[Proposition 1]{bwz}. 
\end{proof}

\begin{remark}\label{rem:manyconnectedcomponents}
The work of Helffer and Sjöstrand \cite[Section 2]{HS88} (as well as the work in \cite{bwz}) shows that the construction in Lemmas \ref{lem:perturbation}--\ref{lem:boundedsymbol} together with Proposition \ref{prop:samespectrum} also applies in the more general setting where $\mu^{-1}(E)$ is a countable union $\cup_\alpha \Gamma_\alpha$ of connected sets $\Gamma_\alpha$, provided these sets are separated by barriers where $\lvert\det(H_0-E_0)\rvert\ge c>0$. This happens, for example, when $H_0$ is periodic (in $x$ and/or $\xi$), and each fundamental domain contains a curve $\gamma=\Gamma_\alpha$ for which our other hypotheses are valid. In the non-periodic case when the $\Gamma_\alpha$'s may be different (or if $H_0$ is periodic but there is more than one connected component in each fundamental domain), the construction of the modified operators $\widetilde H_\alpha^w$ corresponding to $\widetilde H^w$ above would generally depend on the set $\Gamma_\alpha$. Provided such modified operators can be constructed, a statement like Proposition \ref{prop:samespectrum} would hold with $Y=\cup_\alpha\spec(\widetilde H_\alpha^w)\cap[E_-,E_+]$. To avoid having our main focus obfuscated by too many technicalities, we have elected to keep things simple and assume that $\mu^{-1}(E)$ just has one connected component.
\end{remark}

From now on we work with $\widetilde H$ and drop the tilde and superscript $Q$ from the notation of $\widetilde H$, $\widetilde H_0^Q$, and all $\widetilde H_j$, as well as from the modified eigenvalues $\widetilde\lambda_\pm$ and the modified spherical coordinates $\widetilde\theta,\widetilde\phi$ used in the proof of Lemma \ref{lem:boundedsymbol}. In other words, we simply write $H_0$ for the principal symbol modified using Lemmas \ref{lem:perturbation} and \ref{lem:boundedsymbol}. 
Then $H$ satisfies the conditions in the following statement.

\begin{thm}\label{thm:Harper2}
Let $M\sim\sum h^j M_j$, where $M_j\in S(1)$ takes values in the space of $n\times n$ Hermitian matrices, and assume there is a unitary  $U_0\in S(1)$ such that
$$
U_0^*M_0U_0=\begin{pmatrix}A_{11}& \\ & A_{22}\end{pmatrix}
$$
for some $k\times k$ matrix-valued $A_{11}$ and $(n-k)\times(n-k)$ matrix-valued $A_{22}$. If $A_{11}(x,\xi)$ and $A_{22}(x,\xi)$ have disjoint sets of eigenvalues, uniformly for $(x,\xi)\in T^*\R^d$, then there is a unitary $U^w\sim\sum h^j U_j^w$ with $U_j\in S(1)$ such that 
$$
(U^w)^*M^wU^w=\begin{pmatrix}A_{11}^w& \\ & A_{22}^w\end{pmatrix}+hD^w+\mathcal \mathcal O_{L^2 \to L^2}(h^\infty)
$$
where $D=\operatorname{diag}(D_{11},D_{22})\in S(1)$ is block-diagonal with $k\times k$ matrix-valued $D_{11}$ and $(n-k)\times(n-k)$ matrix-valued $D_{22}$ having principal symbols
$$
\Big[U_1^*M_0U_0+U_0^*M_0U_1+  U_0^*M_1U_0+\frac1{2i}\Big(U_0^*\{M_0,U_0\}+\{U_0^*,M_0U_0\}\Big)\Big]_{jj},
$$
for $j=1,2$. 
\end{thm}

\begin{proof}
This is just a reformulation of \cite[Proposition 3.1.1]{HS90} and \cite[Corollary 3.1.2]{HS90} for the Weyl quantization, and follows by using Taylor's trick \cite{taylor1975reflection} adapted to semiclassical operators with an additional argument to make $U^w$ unitary. The last formula follows directly from the Weyl calculus. 
\end{proof}

We apply the theorem to $H$ with $A_{11}=\mu$, where $\mu\in\{\lambda_+,\lambda_-\}$ is the eigenvalue satisfying \eqref{eq:curve}, and $A_{22}$ is the other one. Since $U^w$ is unitary, and $(U^w)^*=(U^*)^w$, the subprincipal symbol of $I=(U^*)^wU^w$ has to vanish. In view of the Weyl calculus, this means that
\[
U_1^*U_0+U_0^*U_1+\frac{1}{2i}\{U_0^*,U_0\}=0.
\]
Hence,
\[
(U_1^*H_0U_0+U_0^*H_0U_1)_{11}=\mu(U_1^*U_0+U_0^*U_1)_{11}=-\frac{\mu}{2i}\{U_0^*,U_0\}_{11}.
\]
We define $U_0$ using a pair of orthonormal eigenvectors, and let $e$ be the eigenvector corresponding to $\mu$ so that the first column of $U_0$ is $e$. By the theorem, we then have
\begin{equation*}
((U^w)^* H^w U^w)_{11} =\mu^w+hf_{1}^w+\mathcal O_{L^2 \to L^2}(h^2),
\end{equation*}
where $f_{1}$ is the principal symbol of $D_{11}$:
\begin{equation}\label{eq:f1}
\begin{aligned}
f_{1}&=\Big(U_0^*H_1U_0+\frac1{2i}\Big(U_0^*\{H_0,U_0\}+\{U_0^*,H_0U_0\}-\mu\{U_0^*,U_0\}\Big)\Big)_{11}
\\&=(U_0^*H_1U_0)_{11}+\mu_1.
\end{aligned}
\end{equation}
Here $\mu_1$ is the top left entry of the subprincipal symbol of $(U^w)^*H_0^wU^w$.
Since the top row of $U_0^*$ is $\bar e$, and the first column of a matrix product $AB$ is $A$ times the first column of $B$, we get
\begin{equation}\label{eq:subpfromH1}
(U_0^*H_1U_0)_{11}=\bar e\cdot H_1e=\langle H_1 e,e\rangle,
\end{equation}
where $\langle\phantom{i},\phantom{i}\rangle$ is the inner product in $\C^2$. Similarly,
\begin{align*}
\mu_1&=\frac{1}{2i}\Big[\bar e\cdot ((H_0)_\xi'e_x'-(H_0)_x'e_\xi')+\{\bar e,\mu e\}-\mu\{\bar e,e\}\Big]
\\&=\frac{1}{2i}\Big[\big\langle \{H_0,e\},e\big\rangle +\mu_x'\langle  e, e_\xi'\rangle-\mu_\xi'\langle e, e_x'\rangle\Big].
\end{align*}
Differentiating $\langle e,e\rangle=1$ gives
\begin{equation*}
\langle e_x',e\rangle=-\langle e,e_x'\rangle,\qquad \langle e_\xi',e\rangle=-\langle e,e_\xi'\rangle,
\end{equation*}
so $\mu_1=\frac{1}{2i}\big\langle \{H_0+\mu,e\},e\big\rangle $. 
We rewrite this as
$$
\mu_1=\frac{1}{2i}\big\langle \{H_0-\mu,e\},e\big\rangle+\frac{1}{i}\big\langle \{\mu,e\},e\big\rangle ,
$$
which is formula \eqref{eq:subprincipalterm} from the introduction. 
Accordingly, we denote the two terms on the right by $\mu_1'$ and $\mu_1''$. This also recovers formula $(6.2.19)$ in \cite{HS90} using the Weyl calculus instead of the calculus for the standard quantization.
By formula $(6.2.30)$ in \cite{HS90} it follows that the curve integrals of $\mu_1'$ and $\mu_1''$ are intrinsically defined in terms of the spectral projection $\Pi$ associated to $\mu$ and don't depend on the choice of eigenvector. (We discuss this in greater detail in Lemmas \ref{lemm:B_invariance} and \ref{lemm:RW_invariance} below.) 
%{\color{blue} In fact, let $\Pi v=\langle v,e\rangle e$. 
%Since $d\Pi = \Pi_x' dx + \Pi_\xi' d\xi$, we have that 
%\[ d\Pi \wedge d\Pi = [\Pi_x',\Pi_\xi'] \,dx \wedge d\xi\]
%and thus 
%\[\tr(\Pi (d\Pi \wedge d\Pi ) \Pi ) = \tr(\Pi ( d\Pi \wedge d\Pi )) = \tr(\Pi [\Pi_x',\Pi_\xi']) \,dx \wedge d\xi.\]
%The quantity
%\begin{equation}\label{eq:Berrycurvature1}
%i\tr(\Pi[\Pi_x',\Pi_\xi'])\,dx \wedge d\xi=-2\im\langle e_x',e_\xi'\rangle\,dx\wedge d\xi
%\end{equation}
%is the Berry curvature. }

\begin{thm}\label{thm:main}
Let $\mu_1=\mu_1'+\mu_1''$ be as in \eqref{eq:subprincipalterm}. Let $H$ be modified using Lemmas \ref{lem:perturbation} and \ref{lem:boundedsymbol} if necessary. Then
$\spec(H^w(x,hD))\cap [E_-,E_+]$ is described by the Bohr--Sommerfeld rule \eqref{eq:BS}, where $f_1$ in \eqref{eq:S1} is given by
$$
f_1=r_0\pm\sum_{i=1}^3\frac{r_ip_i}{\lVert P\rVert}+\mu_1.
$$
\end{thm}

\begin{proof}
In view of Theorem \ref{thm:Harper2} we have for unitary $U^w$ that
$$
(U^w)^*H^w U^w=\begin{pmatrix}\mu^w& \\ & A_{22}^w\end{pmatrix}+\mathcal O_{L^2}(h),
$$
where $A_{22}^{-1}(E)=\emptyset$ for all $E\in I$ by Lemma \ref{lem:boundedsymbol}. Hence, for $h$ small we have that 
\begin{equation}\label{eq:specidentity}
\spec(H^w(x,hD))\cap I=\spec((U^w)^*H^w(x,hD) U^w)_{11}\cap I,
\end{equation}
where $((U^w)^*H^w U^w)_{11}$ has principal symbol $\mu$ and subprincipal symbol $f_1$. Since $d\mu\ne0$ near $\gamma=\mu^{-1}(E)$, we find by \cite[Theorem 8.4]{sjostrand2006cime} and \cite{cdv} that \eqref{eq:specidentity} is given by a Bohr--Sommerfeld rule of the form \eqref{eq:BS}, where $f_1$ in \eqref{eq:S1} is given by \eqref{eq:f1}. Combining this with \eqref{eq:subpfromH1} and \eqref{eq:subprincipalterm} we see that the result follows if we show that
\begin{equation}\label{eq:H1formula}
\langle H_1 e,e\rangle = r_0\pm\sum_{i=1}^3\frac{r_ip_i}{\lVert P\rVert}
\end{equation}
when $e$ is an eigenvector corresponding to $\mu=\lambda_\pm$. To this end, let $v=(v_1,v_2,v_3)$ with $v_i=\langle \sigma_i e,e\rangle$. Then $\langle H_1 e,e\rangle = r_0+\sum_{i=1}^3r_iv_i$. A standard calculation shows that $\lVert v\rVert=1$ for any normalized vector $e$. Now, $p_0\pm\lVert P\rVert=\mu=\langle H_0e,e\rangle$, so
$$
\lVert P\rVert=\pm\sum_{i=1}^3p_i\langle \sigma_ie,e\rangle=\pm P\cdot v.
$$
Hence, $v=\pm P/\lVert P\rVert$, and the result follows.
\end{proof}

We can now prove Theorem \ref{thm:main0} from the introduction.

\begin{proof}[Proof of Theorem \ref{thm:main0}]
For clarity, let $H^w(x,hD)$ be the original unperturbed system, and write $\widetilde H^w(x,hD)$ for the system modified using Lemmas \ref{lem:perturbation} and \ref{lem:boundedsymbol}. The principal symbols of $H^w$ and $\widetilde H^w$ are $H_0$ and $\widetilde H_0^Q$. By Theorem \ref{thm:main}, the spectrum of $\widetilde H^w$ in $[E_-,E_+]$ is given by the Bohr--Sommerfeld rule \eqref{eq:BS}, where $f_1$ and $\mu_1$ are as in the statement of Theorem \ref{thm:main0} since $\widetilde H_0^Q=H_0$ near $\gamma$. Combining this with Proposition \ref{prop:samespectrum} gives \eqref{eq:main0}, and the result follows. 
\end{proof}

\section{The subprincipal term}\label{sec:subprincipal}

Here we describe $\mu_1=\mu_1'+\mu_1''$ in terms of spherical coordinates $\theta$ and $\phi$ in \eqref{eq:sphericalcoordinates}. To the eigenvalues $\lambda_\pm$ we choose corresponding eigenvectors $u_\pm$ of $H_0$ in \eqref{eq:easyH}, which are independent of $p_0$ and well-defined when $\lVert P \rVert \neq 0$, such that
\begin{equation}\label{eq:EVs}
u_+ = \begin{pmatrix} \cos(\theta/2) \\ e^{i\phi} \sin(\theta/2) \end{pmatrix}, \quad
u_- = \begin{pmatrix} -e^{-i\phi} \sin(\theta/2) \\ \cos(\theta/2) \end{pmatrix}.
\end{equation}
We then define $U_0$ in \eqref{eq:f1} using the eigenvectors $u_\pm$, and let $e\in\{u_+,u_-\}$ be the eigenvector corresponding to $\mu=\lambda_\pm$.

\begin{lemma}\label{lem:RWphase}
 Let $\mu_1'=\frac{1}{2i}\langle\{H_0-\lambda_\pm,u_\pm\},u_\pm\rangle$. Then for both $\pm$ we have
$$
\mu_1'=\lVert P\rVert  \frac{\sin(\theta)}2\{\theta,\phi\},
\qquad
\mu_1'\, dx\wedge d\xi=-\lVert P\rVert  \frac{\sin(\theta)}2\,d\theta\wedge d\phi.
$$
\end{lemma}

\begin{proof}
Write $\mu=\lambda_\pm$ and $e=u_\pm$. We differentiate $(H_0-\mu)e=0$ and obtain
\begin{align}
\label{eq:diffid1}
((H_0)_x'-\mu_x')e&=-(H_0-\mu)e_x',\\
((H_0)_\xi'-\mu_\xi')e&=-(H_0-\mu)e_\xi'.
\label{eq:diffid2}
\end{align}
In view of \eqref{eq:diffid1}--\eqref{eq:diffid2} we see that
$$
\langle\{H_0-\mu, e\},e\rangle=\langle e_\xi',(H_0-\mu)e_x'\rangle-\langle e_x',(H_0-\mu)e_\xi'\rangle
$$
and since $H_0-\mu$ is self-adjoint this gives
$$
\mu_1'=-\im \langle (H_0-\mu)e_x',e_\xi'\rangle.
$$
A straightforward calculation using \eqref{eq:EVs} then shows that 
$$
\mu_1'=\lVert P\rVert  \frac{\sin(\theta)}2\{\theta,\phi\}
$$
when $e=u_\pm$. 
Since
$$
d\theta\wedge d\phi=(\theta_x'dx+\theta_\xi'd\xi)\wedge (\phi_x'dx+\phi_\xi'd\xi)=\{\phi,\theta\}\,dx\wedge d\xi
$$
we have that
$$
\mu_1'\, dx\wedge d\xi=-\lVert P\rVert  \frac{\sin(\theta)}2\,d\theta\wedge d\phi,
$$
which completes the proof.
\end{proof}

\begin{lemma}\label{lem:Berryphase} 
Let $\mu_1''=\frac{1}{i}\langle\{\lambda_\pm,u_\pm\},u_\pm\rangle$. 
Then 
$$
\mu_1''=\pm\frac{1-\cos(\theta)}2\{\lambda_\pm,\phi\},
\qquad
\int_\gamma \mu_1''\,dt=\pm\int_{\gamma}\frac{1-\cos(\theta)}2\,d\phi.
$$
\end{lemma}

\begin{proof}
It is easy to check that $\langle e_x',e\rangle=\pm i \phi_x'(1-\cos(\theta))/2$ when $e=u_\pm$, and a similar formula holds for the $\xi$-derivative. Writing $\mu=\lambda_\pm$, this gives
$$
\mu_1''=\frac{1}{i}(\mu_\xi'\langle e_x',e\rangle-\mu_x'\langle e_\xi',e\rangle)=\pm\frac{1-\cos(\theta)}2\{\mu,\phi\}.
$$
Since $dx=\mu_\xi'\,dt$ and $d\xi=-\mu_x'\,dt$ on $\gamma$ it follows that if $\alpha_\pm$ is the 1-form $$\alpha_\pm=\pm\frac{1-\cos(\theta)}2\,d\phi=\pm\frac{1-\cos(\theta)}2(\phi_x'\,dx+\phi_\xi'\,d\xi)$$ then $\alpha_\pm$ restricted to $\gamma$ is equal to $\mu_1''(x(t),\xi(t))\,dt$, which proves the second identity of the statement.
\end{proof}

%{\color{blue}
%\begin{lemma}\label{lem:auxphasecalculation}
%$\displaystyle \Im \langle \partial_x u_{\pm} ,\partial_{\xi} u_{\pm} \rangle = \pm\frac14\sin(\theta)\{\theta,\phi\}$.
%\end{lemma}
%
%\begin{proof}
%Write $e=u_\pm$. By the chain rule we have
%$$
%\langle e_x' ,e_\xi' \rangle =\theta_x'\theta_\xi' \lVert e_\theta'\rVert^2+\phi_x'\phi_\xi' \lVert e_\phi'\rVert^2+\phi_x'\theta_\xi'\langle e_\phi',e_\theta'\rangle+\phi_\xi'\theta_x'\langle e_\theta',e_\phi'\rangle.
%$$
%Since $\im\langle e_\theta',e_\phi'\rangle=-\im \langle e_\phi',e_\theta'\rangle$ we get
%$$
%\Im \langle e_x' ,e_\xi' \rangle =\{\theta,\phi\}\im \langle e_\phi',e_\theta'\rangle.
%$$
%A straightforward calculation using \eqref{eq:EVs} shows that $\im \langle e_\phi',e_\theta'\rangle=\pm\frac14\sin(\theta)$ when $e=u_\pm$, which completes the proof.
%\end{proof}
%}

\begin{remark}[Berry's phase]\label{rem:Berrysphase}
The Berry connection is $A_{\pm} := i \langle u_{\pm},du_{\pm} \rangle$. By  the chain rule we get
\[ A_\pm =  i \langle u_{\pm},\partial_r u_{\pm} \rangle dr +  i \langle u_{\pm},\partial_{\theta} u_{\pm} \rangle d\theta + i  \langle u_{\pm},\partial_{\phi} u_{\pm} \rangle d\phi.\]
In view of \eqref{eq:EVs} this implies 
\[A_\pm =\pm \frac{1-\cos(\theta)}{2} d\phi.\]
The Berry phase is thus
\begin{equation*}
\theta_{B} =\pm \int_{\gamma}  \frac{1-\cos(\theta)}{2} d\phi,
\end{equation*}
which is precisely $\int_\gamma \mu_1'' \,dt$ in view of Lemma \ref{lem:Berryphase}.
The Berry curvature associated with the Berry connection $A_\pm$ is then 
\begin{equation}\label{eq:Berrycurvature2}
F_\pm =dA_\pm = \pm \frac{\sin(\theta)}{2}\, d \theta \wedge d\phi.
\end{equation}
%{\color{blue} By Lemma \ref{lem:auxphasecalculation} we have $\im\langle e_x' ,e_\xi' \rangle=\pm\frac14\sin(\theta)\{\theta,\phi\}$ which shows that this expression for the Berry curvature is the same as \eqref{eq:Berrycurvature1}.}
\end{remark}
The above characterization of the $\mu_1''$ contribution as the Berry phase implies the following, of course, well-known lemma which is useful in practical computations. 
\begin{lemma}
\label{lemm:B_invariance}
Let $v,w$ be any smooth linearly dependent normalized vectors in a neighbourhood of a simple closed smooth curve $\gamma$, then their Berry phases agree up to a term in $2\pi \mathbb Z.$
\end{lemma}
\begin{proof}
Linear dependence implies that $v = e^{i\eta} w$ for some smooth $e^{i\eta}.$
This implies that the Berry connections $A(u):=i \langle u, du\rangle$ are related by 
\[ A(v)=A(w) + d\eta.\]
Integrating implies the claim. 
\end{proof}

\begin{remark}[Rammal--Wilkinson's phase]\label{rem:RWphase}
Set $\mathcal A_\pm=i\langle u_\pm,du_\mp\rangle$. Expanding the $1$-forms we have
\[\begin{split}
\mathcal A_+&=:A_+^{(x)}\,dx+A_+^{(\xi)}\,d\xi \\
\mathcal A_-&=:A_-^{(x)}\,dx+A_-^{(\xi)}\,d\xi.
\end{split}\]
Then a computation shows that
\[-2\Im(A_\mp^{(x)}A_\pm^{(\xi)})=-\frac12\im\big[(\sin(\theta)\phi_x'\pm i\theta_x')(\sin(\theta)\phi_\xi'\mp i\theta_\xi')\big]=\pm\frac{\sin(\theta)}{2}\{\theta,\phi\}.\]
As we have seen, the right-hand side is the Berry curvature scalar density in phase space coordinates (cf.~\eqref{eq:Berrycurvature2}). Let's denote it by $f_\pm$, then $\mu_1'=\pm\lVert P\rVert f_\pm$ by Lemma \ref{lem:RWphase}, so
\begin{equation}
\label{eq:RW_connection}
\mu_1'= -2 \Vert P \Vert \Im(A_-^{(x)}A_+^{(\xi)}).
 \end{equation}
\end{remark}
Similar to the case of the Berry phase, we get the following rigidity result for the Rammal--Wilkinson phase.
\begin{lemma}
\label{lemm:RW_invariance}
Let $v_+,w_+$ be any smooth linearly dependent normalized vectors with orthogonal linearly dependent normalized vectors $v_-,w_-$ in a neighbourhood of a simple closed smooth curve $\gamma$, then their Rammal--Wilkinson phases computed from \eqref{eq:RW_connection} coincide.
\end{lemma}
\begin{proof}
Linear dependence implies that $v_{\pm} = e^{i\eta_{\pm}} w_{\pm}$ for some smooth $e^{i\eta_{\pm}}.$
This implies by orthogonality that 
\[\mathcal A_\pm(v) = i \langle v_{\pm},dv_{\mp}\rangle = i e^{i(\eta_{\pm}-\eta_\mp)} \langle w_{\pm}, dw_{\mp} \rangle = e^{i(\eta_{\pm}-\eta_\mp)} \mathcal A_{\pm}(w).\]
This shows invariance by noticing that
\begin{align*} 
\mu_1' &= -2\Vert P \Vert\Im(A_-^{(x)}(v)A_+^{(\xi)}(v))\\
&= -2 \Vert P \Vert\Im(A_-^{(x)}(w)A_+^{(\xi)}(w)). \qedhere\end{align*} 
\end{proof}

The results of this section now combine into a proof of Theorem \ref{thm:main1}.

\begin{proof}[Proof of Theorem \ref{thm:main1}]
From \eqref{eq:subprincipalterm} and Lemmas \ref{lem:RWphase}--\ref{lem:Berryphase} we conclude that if $\mu=\lambda_\pm$ then
\begin{equation}\label{eq:finalformula}
\mu_1=\lVert P\rVert  \frac{\sin(\theta)}2\{\theta,\phi\}\pm\frac{1-\cos(\theta)}2\{\lambda_\pm,\phi\}.
\end{equation}
The stated formulas for $\theta_B$ and $\theta_{\mathit{RW}}$ now follow in view of Remarks \ref{rem:Berrysphase} and \ref{rem:RWphase}. 

When calculating the formulas we have used eigenvectors $u_\pm$ from \eqref{eq:EVs} that are well-defined and smooth where $P\ne0$, which holds by assumption near $\gamma$.  To see that this is justified, assume $P\ne0$ doesn't hold globally. We then take a small $\ve>0$ and modify $H_0$ using Lemma \ref{lem:perturbation} and replace $P$ with some $Q$ such that $\sup_{w\in\{Q(w)\ne P(w)\}}\lVert Q(w)\rVert<\ve$. This modified Hamiltonian now has gapped eigenvalues and allows for a smooth choice of eigenvectors. 

When integrating the Berry connection or the Berry curvature scalar density of these eigenvectors over an energy level curve $E>0$, then for $\varepsilon>0$ small enough, the corresponding smooth eigenvector can be replaced by a smooth section of eigenvectors on that energy level curve of the unperturbed system. 
More precisely, there is the smooth choice of eigenvectors of the perturbed system $v_{\varepsilon}$ and the smooth choice of the unperturbed system $v$, i.e., by $u_\pm$ from \eqref{eq:EVs}. On an energy level set $\gamma$ they agree up to a continuous phase factor $e^{i \eta}: \gamma \to \mathbb S^1,$ the complex unit circle. It follows from Lemmas \ref{lemm:B_invariance} and \ref{lemm:RW_invariance} that the two computations agree up to a term $2\pi n$ for $n \in \mathbb Z$ which does not affect the Bohr--Sommerfeld rule. This also shows the invariance statement in Theorem \ref{thm:main1} and the proof is complete.
\end{proof}

\section{Quantized geometric phase corrections}\label{sec:phasecalculations}

Let $\operatorname{wind}(\Gamma,0)$ be the winding number of the curve $\Gamma\subset \C$ around the origin in the complex plane. We first establish that the integral $\int_\gamma\mu_1\,dt$ is quantized when $p_3$ vanishes identically.

\begin{prop}\label{prop:quantizedmu1}
Let $\mu=\lambda_\pm$, and assume that $p_3\equiv0$. Let $\Gamma=(p_1+ip_2)(\gamma)\subset\C$ be the image of $\gamma$ under $p_1+ip_2:T^*\R\to \C$ with the induced orientation. Then 
$$
\int_\gamma \mu_1\, dt=\theta_B=\pm\pi\operatorname{wind}(\Gamma,0),
$$
where $\theta_B$ is the Berry phase.
\end{prop}

\begin{proof}
If $p_3\equiv0$ then the spherical coordinate $\theta$ in \eqref{eq:sphericalcoordinates} is constant, $\theta\equiv\pi/2$, which means that $\{\theta,\phi\}\equiv0$. By \eqref{eq:finalformula} and Lemma \ref{lem:Berryphase}  we then get $\int_\gamma\mu_1\,dt=\pm\frac12\int_\gamma d\phi$. Since 
$$
d\phi=\frac{p_1dp_2-p_2dp_1}{p_1^2+p_2^2}
$$
the result follows by a change of variables.
\end{proof}

Next, we record the following special case which applies to a certain Dirac operator that is discussed in \S\ref{ss:TM} in connection to strained moiré lattices.

\begin{prop}\label{prop:quantizedH1}
Let $\mu=\lambda_\pm$, and assume that $H\sim H_0+hH_1+\ldots$ with $H_1= k\sigma_2$ for some constant $k\in\R$. 
If $p_j=p_j(x)$ for $j=0,1$, while $p_2=p_2(\xi)=\xi$ and $p_3\equiv0$, then 
$$
\int_\gamma \langle H_1e,e\rangle\, dt=0.
$$
The same is true if $H_1=k\sigma_1$, $p_j=p_j(\xi)$ for $j=0,2$, while $p_1=p_1(x)=x$ and $p_3\equiv0$.
\end{prop}

\begin{proof}
By \eqref{eq:H1formula} we have that if $\mu=\lambda_\pm$ and $H_1=k\sigma_2$ then 
$$
\langle H_1e,e\rangle=\pm k\frac{p_2}{\sqrt{p_1^2+p_2^2}}.
$$
Next, we note that if $p_j=p_j(x)$ for $j=0,1$ while $p_2=\xi$ then
$$
\frac{\partial\lambda_\pm}{\partial\xi}=\pm \frac{\xi}{\sqrt{p_1^2+\xi^2}}\quad\Longrightarrow\quad \langle H_1e,e\rangle=k\frac{\partial\lambda_\pm}{\partial\xi}.
$$
Since $dx/dt=\partial\lambda_\pm/\partial\xi$ on $\gamma$ we get $\int_\gamma \langle H_1e,e\rangle \,dt=k\int_\gamma dx$, so the conclusion follows by Stokes' theorem.

If instead $H_1=k\sigma_1$, then the corresponding assumptions imply by similar arguments that
$$
\langle H_1e,e\rangle=\pm k\frac{p_1}{\sqrt{p_1^2+p_2^2}}=k\frac{\partial\lambda_\pm}{\partial x},
$$
and since $d\xi/dt=-\partial\lambda_\pm/\partial x$ on $\gamma$ we get $\int_\gamma \langle H_1e,e\rangle \,dt=-k\int_\gamma d\xi=0$.
\end{proof}

\begin{cor}\label{cor:quantizedS1}
Assume, in addition to the hypotheses of Proposition \ref{prop:quantizedmu1}, that $H\sim H_0+hH_1+\ldots$ and that either $H_1\equiv0$, or that the assumptions of Proposition \ref{prop:quantizedH1} are in force. Then 
$$
S_1(E)=\pi-\theta_B=\pi\mp\pi\operatorname{wind}(\Gamma,0).
$$
\end{cor}

\begin{proof}
This follows immediately from Theorem \ref{thm:main} and Propositions \ref{prop:quantizedmu1} and \ref{prop:quantizedH1}.
\end{proof}

Corollary \ref{cor:quantizedS1} means that $S_1(E)$ becomes quantized when $p_3$ vanishes identically near $\gamma$, and $H_1$ is either trivial or constant in the circumstances described by Proposition \ref{prop:quantizedH1}. We will now show that there is nothing special about $p_3$. In fact, (with the same caveat about $H_1$) it turns out that $S_1(E)$ becomes quantized whenever $P=(p_1,p_2,p_3)$ lies in a plane near $\gamma$, so that the components of $P$ are linearly dependent over $\R$ there.

Indeed, suppose there is a constant vector $C=(c_1,c_2,c_3)\in\R^3$ with $\lVert C\rVert=1$ such that $C\cdot P=0$ near $\gamma$. We can then rotate $P$ to eliminate the $\sigma_3$ component. If $p_3\equiv0$ we don't do anything, otherwise we take a normalized rotational axis vector $$n=\frac{C\times (0,0,1)}{\lVert C\times (0,0,1)\rVert}=\frac1{\sqrt{c_1^2+c_2^2}}(c_2,-c_1,0)$$ which will then be non-trivial. The angle between $C$ and the $z$ axis is $\omega=\arccos\left(c_3\right)$ so we define the unitary matrix
$$
U=\exp(-\frac{i\omega}2 n\cdot\sigma)=\exp(-\frac{i\arccos(c_3)}{2}\cdot \frac{(c_2\sigma_1-c_1\sigma_2)}{\sqrt{c_1^2+c_2^2}}).
$$
Note that
$$
(n\cdot\sigma)^2=\bigg(\frac{c_2}{\sqrt{c_1^2+c_2^2}}\sigma_1-\frac{c_1}{\sqrt{c_1^2+c_2^2}}\sigma_2\bigg)^2=I
$$
so 
$$
U=\cos(\omega/2)I-i\sin(\omega/2)(n\cdot\sigma).
$$
Then
$$
U^*\bigg(\sum_{i=1}^3p_i\sigma_i\bigg)U=\sum_{i=1}^3q_i\sigma_i,
$$
where $Q=(q_1,q_2,q_3)$ is the vector obtained through rotating $P$ by $\omega=\arccos(c_3)$ around $n$.
By Rodrigues' rotation formula we get
$$
Q=P\cos(\omega)+(n\times P)\sin(\omega)+n(n\cdot P)(1-\cos(\omega)),
$$
which gives
$$
Q=Pc_3+(-c_1p_3,-c_2p_3,c_1p_1+c_2p_2)+(c_2,-c_1,0)\frac{(c_2p_1-c_1p_2)(1-c_3)}{c_1^2+c_2^2}.
$$
Hence, $q_3=c_1p_1+c_2p_2+c_3p_3=0$ and
\begin{equation*}%\label{eq:q1q2}
\begin{aligned}
q_1&=p_1\bigg(c_3+\frac{c_2^2(1-c_3)}{c_1^2+c_2^2}\bigg)-p_2\frac{c_1c_2(1-c_3)}{c_1^2+c_2^2}-c_1p_3,\\
q_2&=p_2\bigg(c_3+\frac{c_1^2(1-c_3)}{c_1^2+c_2^2}\bigg)-p_1\frac{c_1c_2(1-c_3)}{c_1^2+c_2^2}-c_2p_3.
\end{aligned}
\end{equation*}
By a simple calculation, this simplifies to
\begin{equation}\label{eq:q1q2}
q_1=p_1-\frac{c_1}{1+c_3}p_3,\qquad
q_2=p_2-\frac{c_2}{1+c_3}p_3.
\end{equation}
Since $U$ is constant and unitary, the Bohr--Sommerfeld rule for $H^w$ is the same as the one for $U^*H^wU$. In particular, we immediately obtain the following generalization of Proposition \ref{prop:quantizedmu1}.

\begin{thm}\label{thm:quantizedmu1}
Let $\mu=\lambda_\pm$, and assume that there is a constant vector $C=(c_1,c_2,c_3)\in\R^3$ with $\lVert C\rVert=1$ such that $C\cdot P=0$ near $\gamma$. If $|c_3|=1$ let $q_1=p_1$ and $q_2=p_2$, whereas if $|c_3|<1$ let $q_1,q_2$ be defined by \eqref{eq:q1q2}. Let $\Gamma=(q_1+iq_2)(\gamma)\subset\C$ be the image of $\gamma$ under $q_1+iq_2:T^*\R\to \C$. Then 
$$
\int_\gamma \mu_1\, dt=\theta_B=\pm\pi\operatorname{wind}(\Gamma,0),
$$
where $\theta_B$ is the Berry phase. In particular, with $\mu_1=\mu_1'+\mu_1''$ as in \eqref{eq:subprincipalterm} we have $\mu_1'=0$.
\end{thm}

We note that Theorem \ref{thm:main2} from the introduction is just a special case of Theorem \ref{thm:quantizedmu1}.
From Theorem \ref{thm:quantizedmu1} we also see that $H_0$ will always have trivial Berry curvature (thus leading to trivial Rammal--Wilkinson phase) if the components of $P=(p_1,p_2,p_3)$ are linearly dependent over $\R$. To find an example of non-trivial Rammal--Wilkinson phase we must therefore look for a system where $P$ has linearly independent components. A simple example having both non-trivial Rammal--Wilkinson phase and non-trivial Berry phase is provided in \S\ref{ss:RMphase}.

\section{Some illustrative examples}\label{sec:examples}
\subsection{The Jackiw-Rebbi model}
\label{sec:JR}
The Jackiw-Rebbi Hamiltonian is of the form
\[ H_{\mathrm{JR}}(x,hD_x) = hD_x \sigma_1 + m(x) \sigma_2,\]
where $m$ is an odd monotonically increasing function with $m^{-1}(0)=0$ and $\lim_{x \to \infty} m(x)=m_0>0$ and $\lim_{x \to \infty} m'(x)=0.$
This implies that 
\begin{lemma}\label{lem:discretespectrum}
The spectrum of $H_{\mathrm{JR}}$ inside $[-m_0,m_0]$ is discrete. 
\end{lemma}
\begin{proof}
The squared operator $H_{\mathrm{JR}}^2$ is the diagonal operator 
\[ \begin{split} 
&(-h^2 \partial_{x}^2 + m(x)^2 )\sigma_0 + hm'(x)\sigma_3.
\end{split}\]
Since $(m(x)^2-m_0^2) \sigma_0 + hm'(x)\sigma_3$ is a relatively compact perturbation of $(-h^2\partial_x^2 +m_0^2)\sigma_0$, the result follows.
\end{proof}

The topological nature of the model is manifested in the non-zero Fredholm index of $Q(x,hD_x):=hD_x -im(x)$.  Indeed, we have 
\[ Q(x,hD_x)\varphi=0 \quad \Longrightarrow \quad h\partial_x \varphi + m(x) \varphi=0.\]
This shows that $\varphi(x) \propto \operatorname{exp}\left(- \int_0^{x} \frac{m(\xi)}{h} \ d\xi \right) \in L^2(\mathbb R),$ while the adjoint operator does not admit an $L^2$ zero mode. Thus, the Fredholm index is
\[\operatorname{ind}(Q)=1. \]

We now use Corollary \ref{cor:main0} to obtain a quantization condition. We note that $H_{\mathrm{JR}}$ has symbol as in \eqref{eq:H0} with $p_1=\xi$ and $p_2=m(x)$, so we therefore need to calculate the winding number of $q(x,\xi)=\xi+im(x)$. In view of the Bohr--Sommerfeld rule, the sign of the winding number is not important, so we will compute the winding number of $q$ as we traverse the unit circle \( (x, \xi) = (\cos t, \sin t) \) for \( t \in [0, 2\pi] \).

To this end, consider the path \( \Gamma(t) = q(\cos t, \sin t) \). We will show that \( \Gamma \) winds once clockwise around the origin in \( \mathbb{C} \) as $t$ increases from $0$ to $2\pi$. Note that \( \Gamma \) is a smooth closed loop in \( \mathbb{C} \setminus \{0\} \). Note also that \( m(\cos t) > 0 \) when \( \cos t > 0 \), i.e., for \( t \in (-\pi/2, \pi/2) \), and \( m(\cos t) < 0 \) when \( \cos t < 0 \), i.e., for \( t \in (\pi/2, 3\pi/2) \). The function \( m(\cos t) \) thus transitions from positive to negative and back to positive as \( t \) increases from 0 to \( 2\pi \), while the real part \( \sin t \) traces a full sine wave from 0 to \(1\), back to 0, to \( -1 \), and finally to 0 again.

This behavior ensures that \( \Gamma(t) \) describes a closed loop that encircles the origin. Since the path moves from the right half-plane (positive real part) to the left half-plane (negative real part) and back, while the imaginary part oscillates from positive to negative values, the overall path sweeps around the origin exactly once in the clockwise direction.
Therefore, the winding number of \( \Gamma \) around zero is $-1$. This yields a Bohr--Sommerfeld rule 
\[ S_0(E)=2\pi k h + \mathcal O(h^2)\]
and its comparison with explicit spectral computations is shown in Figure \ref{fig:JR}.

\begin{figure}
\includegraphics[width=6.5 cm]{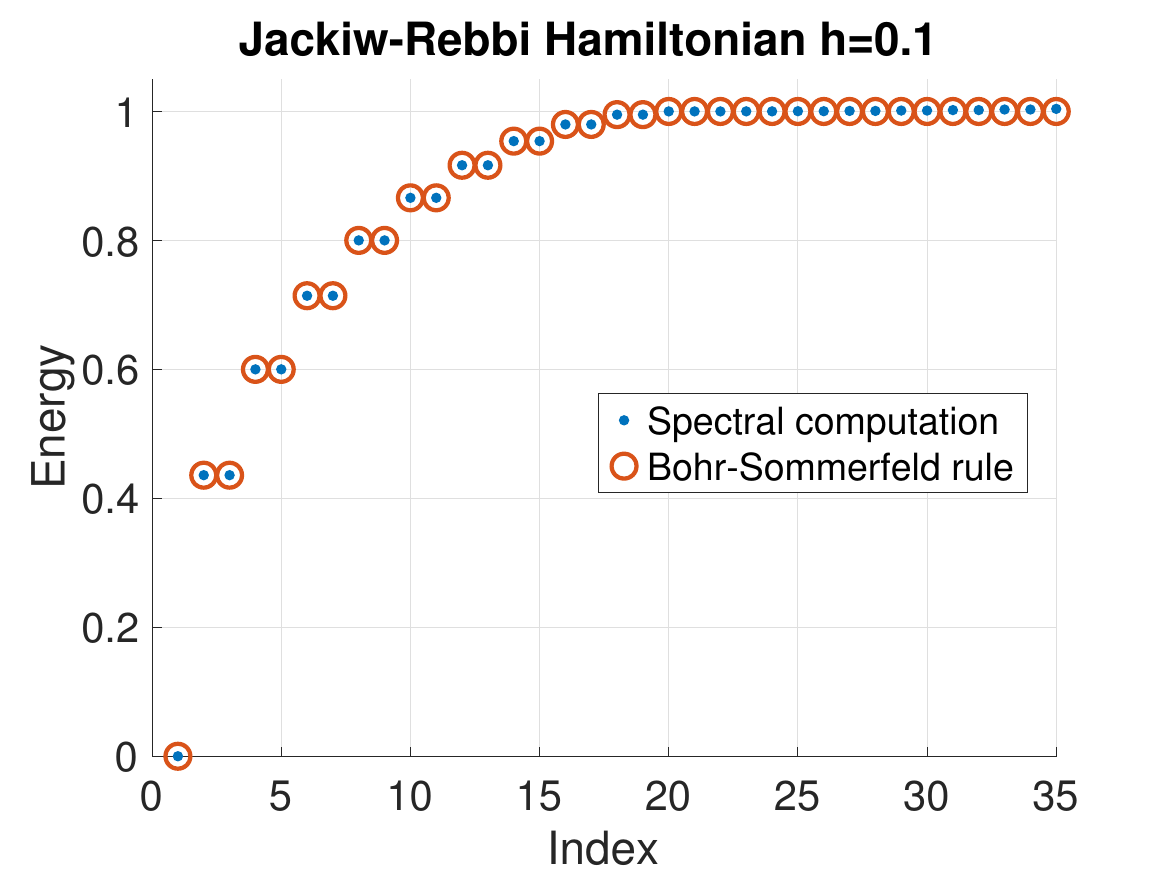}
\includegraphics[width=6.5 cm]{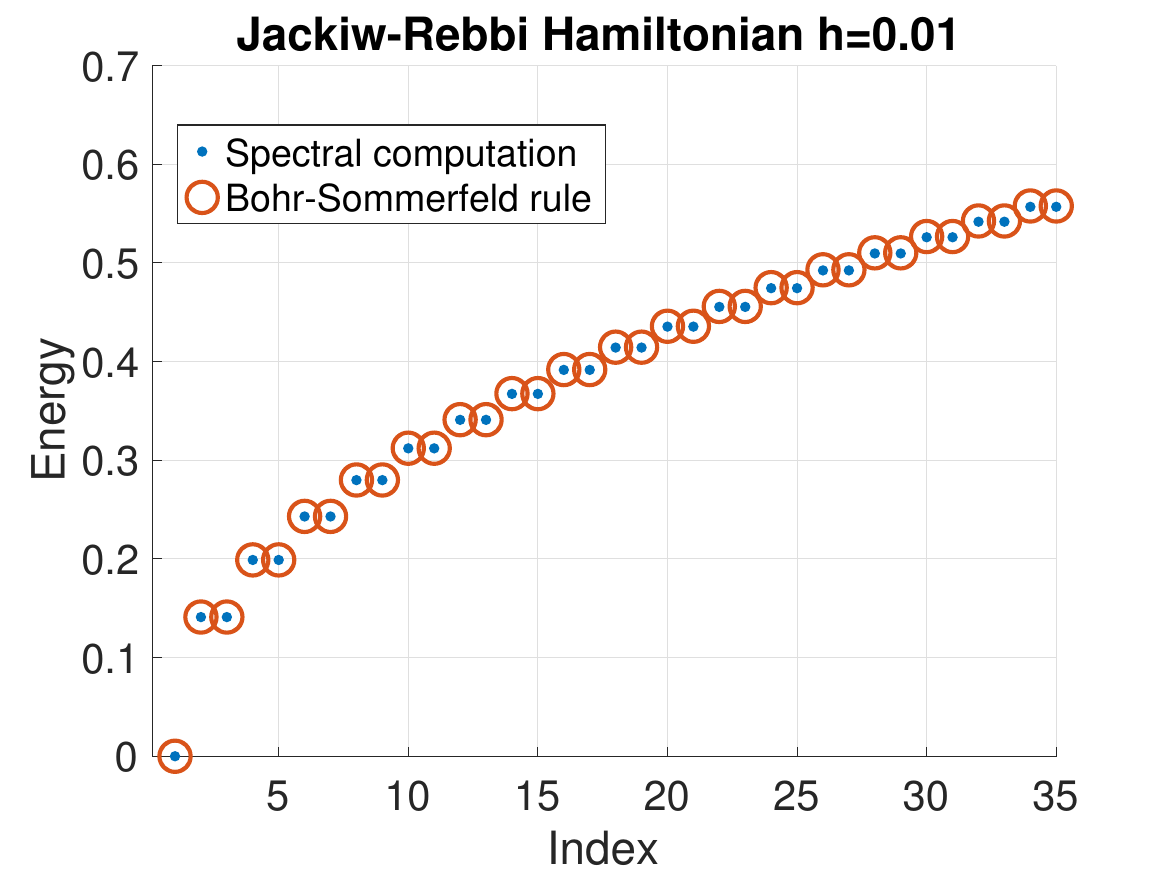}
\caption{Comparison of the smallest 35 absolute values of explicit eigenvalues of Jackiw-Rebbi Hamiltonian with $m(x)=\tanh(x)$ for $h=0.1$ and $h=0.01.$}
\label{fig:JR}
\end{figure}

\subsection{Non-trivial Rammal--Wilkinson phase}\label{ss:RMphase}
Take $p_0=0$ for simplicity, and let $p_1(x,\xi)=x$, $p_2(x,\xi)=\xi$, and $p_3(x,\xi)=x^2$ so that 
\[ H(x,\xi) = \begin{pmatrix} x^2 & x-i\xi \\ x+i\xi& -x^2 \end{pmatrix} \text{ and thus }H^w(x,hD_x) = \begin{pmatrix} x^2 & x-h\partial_x \\ x+h\partial_x & -x^2 \end{pmatrix}.\]
As a self-adjoint (densely defined) operator on $L^2(\R)$, $H^w$ has discrete spectrum. This follows by noting that the square of $H^w$ is equal to $(-h^2\Delta +x^2+x^4)\sigma_0$ modulo lower order terms, which is more confining than the harmonic oscillator, so an argument similar to the proof of Lemma \ref{lem:discretespectrum} proves the claim.

Using \eqref{eq:sphericalcoordinates} it is straightforward to check that $\sin(\theta)=\sqrt{x^2+\xi^2}/\lVert P\rVert$, and that
\begin{equation*}
\frac{\partial\theta}{\partial x}=-\frac{x^3+2x\xi^2}{\sqrt{x^2+\xi^2}\lVert P\rVert^2},\qquad  
\frac{\partial\theta}{\partial \xi}=\frac{x^2\xi}{\sqrt{x^2+\xi^2}\lVert P\rVert^2}
\end{equation*}
and
\begin{equation*}
\frac{\partial\phi}{\partial x}=-\frac{\xi}{x^2+\xi^2},\qquad 
\frac{\partial\phi}{\partial \xi}=\frac{x}{x^2+\xi^2}.
\end{equation*}
This gives
$$
\sin(\theta)\{\theta,\phi\}=\frac{1}{(x^2+\xi^2)\lVert P\rVert^3}(-x^2\xi^2+x(x^3+2x\xi^2))=\frac{x^2}{\lVert P\rVert^3}.
$$

From \eqref{eq:finalformula} we see that
$$
\mu_1'=\frac12\lVert P\rVert\sin(\theta)\{\theta,\phi\}=\frac{x^2}{2(x^2+\xi^2+x^4)}.
$$
Let's consider $\mu=\lambda_+=\sqrt{x^2+\xi^2+x^4}$ and a positive energy level $E>0$. The level curve $\gamma=\mu^{-1}(E)$ is roughly circular, and we have 
$$
\int_\gamma \mu_1'\,dt=\frac1{2E^2}\int_{\{x^2+\xi^2+x^4=E^2\}}x^2\,dt.
$$
On $\gamma$ we have $dx/dt=\partial\mu/\partial\xi=\xi/E$, where $\xi^2=E^2-x^2-x^4$. Using the symmetry of $\gamma$ we can then eliminate $t$ and use $x$ as the integration variable, so that
\begin{equation}\label{eq:RWphase}
\int_\gamma \mu_1'\,dt=\frac2{E}\int_0^{x_0}\frac{x^2}{\sqrt{E^2-x^2-x^4}}\,dx
\end{equation}
where $x_0$ is the turning point
$$
E^2=x_0^2+x_0^4\quad\Longrightarrow\quad x_0=\sqrt{\frac{-1+\sqrt{1+4E^2}}{2}}.
$$
The Rammal--Wilkinson phase \eqref{eq:RWphase} is clearly non-zero, and can be computed numerically as a function of $E$. It is not quantized but depends continuously on $E$.

To compute the Berry phase, we note that
\[ \frac{1-\cos(\theta)}{2} = \frac{\lVert P \rVert-p_3}{2\lVert P \rVert} = \frac{\lVert P \rVert-x^2}{2\lVert P \rVert} ,\]
while $\partial_x \lambda_+(x,\xi)= x(1+2x^2)/\lVert P \rVert$ and $\partial_{\xi} \lambda_+(x,\xi) = \xi/\lVert P \rVert$. 
It follows that 
\[ \{\lambda_+,\phi\} = -\frac{x^2+2x^4+\xi^2}{(x^2+\xi^2)\lVert P \rVert}.\]
This gives
\[ \frac{1-\cos(\theta)}{2} \{\lambda_+,\phi\}=  -\frac{(\lVert P \rVert-x^2)(x^2+2x^4+\xi^2)}{2(x^2+\xi^2)\lVert P \rVert^2}.\]
We can eliminate $\xi$ and get an integral expression for the Berry phase as
\begin{equation}
\label{eq:Berry_phase}
\begin{split} 
\int_{\gamma} \mu_1'' \,dt &=  -2 \int_0^{x_0} \frac{(E-x^2)(E^2+x^4)}{(E^2-x^4)E\sqrt{E^2-x^2-x^4}} \, dx \\
&= - 2 \int_0^{x_0} \frac{E^2+x^4}{E(E+x^2)\sqrt{E^2-x^2-x^4}} \, dx.
\end{split}
\end{equation}

Both phases \eqref{eq:RWphase} and \eqref{eq:Berry_phase} are illustrated in Figure \ref{fig:JR2} together with the resulting Bohr--Sommerfeld rule.

\begin{figure}
\includegraphics[width=6.3 cm]{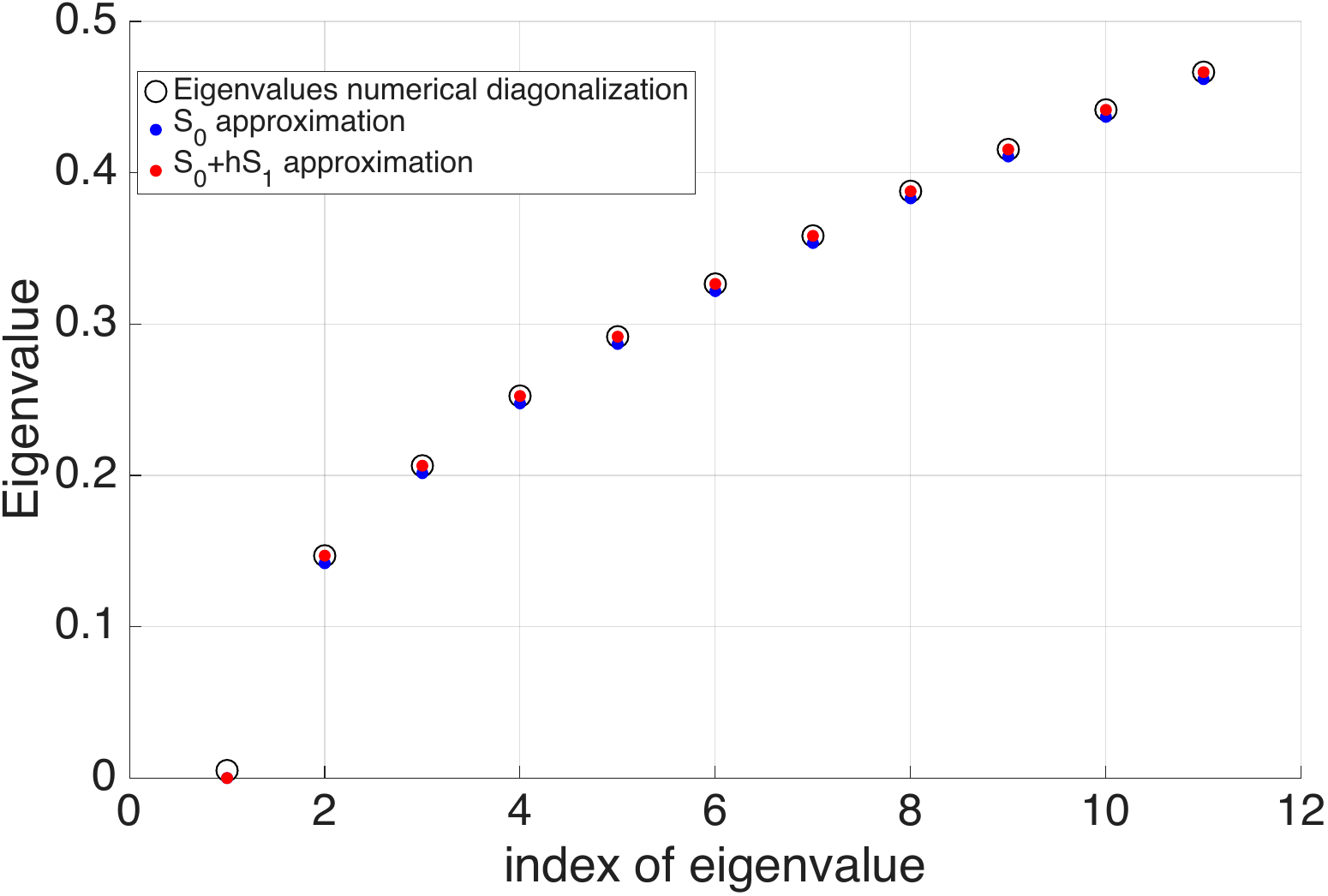}
\includegraphics[width=6.3 cm]{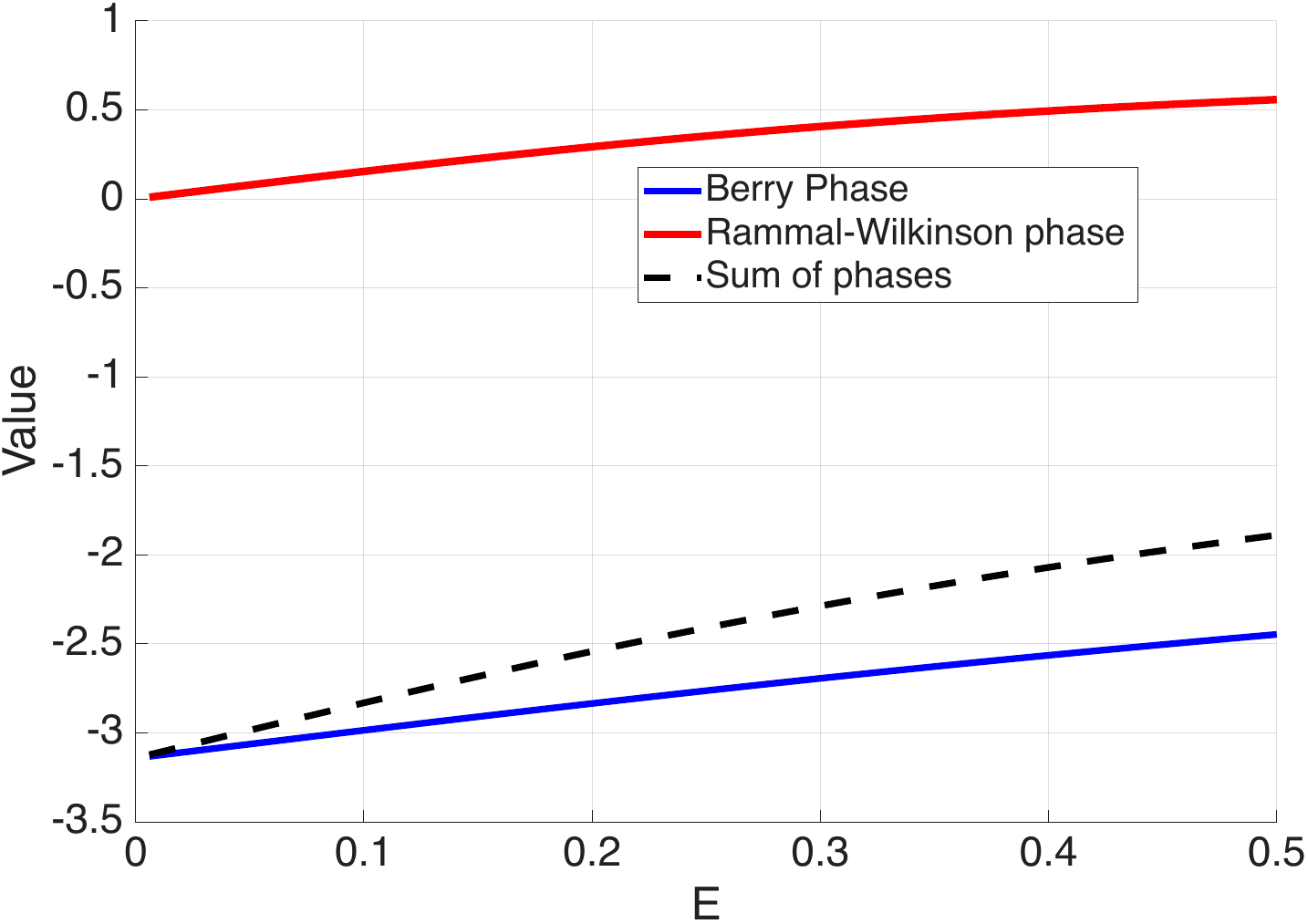}
\caption{Left: Comparison between spectral computation of the 11 smallest positive eigenvalues for $h=0.01$ compared with semiclassical approximation by $S_0$ and $S_0 + hS_1$, respectively. The addition of the $S_1$ term leads to improvement away from zero energy. Right: We plot the non-quantized Berry and Rammal--Wilkinson phases.}
\label{fig:JR2}
\end{figure}

\subsection{Strained moiré lattices}\label{ss:TM}
Here we present the Timmel-Mele model \cite{TM2020} of strained moiré lattices for a choice of parameter settings for which the zero energy level set of the symbol contains simple closed curves. For more general parameter settings, we refer to \cite[Section 4]{becker2022semiclassical}.

The chiral version of the Timmel-Mele model \cite{TM2020} is given by a $4\times4$ Hamiltonian which is unitarily equivalent to 
$$
H_{\mathrm{TM}}=\begin{pmatrix} 0& D_c\\ D_c^* &0\end{pmatrix},
$$
where $D_c$ ($c$ for {\it chiral}) is a non-self-adjoint $2\times2$ system given by
$$
D_c=\frac12\begin{pmatrix} i&1\\ -i &1\end{pmatrix}\begin{pmatrix}\Upsilon^w(hD_x)&U^+(x)\\ U^-(x) &\Upsilon^w(hD_x)\end{pmatrix}\begin{pmatrix} -i&i\\ 1 &1\end{pmatrix}.
$$
Here $U^\pm(x)=1-\cos(2\pi x)\pm\sqrt 3\sin(2\pi x)$, and the symbol $\Upsilon(\xi)$ comes in two variants: one for a tight-binding model where
$$
\Upsilon(\xi)=\Upsilon_0(\xi)=2\cos(2\pi\xi)+1,
$$
and one for a low-energy approximation where
$$
\Upsilon(\xi)=\Upsilon_0(\xi)+h\Upsilon_1(\xi)=\xi+hk_x.
$$
The real parameter $k_x$ is a quasi-momentum coming from a Bloch-Floquet transformation of the original model. In both cases, the corresponding operator acts on $L^2(\mathbb T;\C^4)$, $\mathbb T:=\R/\Z$. (When $\Upsilon(\xi)=\xi$, the operator is then only densely defined.)
Squaring $H_{\mathrm{TM}}$ gives
$$
H_{\mathrm{TM}}^2=\begin{pmatrix} D_cD_c^*&0\\ 0&D_c^*D_c \end{pmatrix},
$$
and $\lambda\ne0$ is in the spectrum of $D_cD_c^*$ if and only if $\pm\sqrt\lambda\ne0$ is in the spectrum of $H_{\mathrm{TM}}$, see \cite[Lemma 4.1]{becker2022semiclassical} (low-energy model) and \cite[Lemma 4.5]{becker2022semiclassical} (tight-binding model). It suffices to study $D_cD_c^*$ since the spectrum of $D_cD_c^*$ and $D_c^*D_c$ are equal away from zero, see Thaller \cite[Corollary 5.6]{thaller2013dirac}.

Introduce the self-adjoint symbols $H_0$ and $H_1$ given by
$$
H_0(x,\xi)=\begin{pmatrix} f(x)& i\Upsilon_0(\xi)-g(x)\\ -i\Upsilon_0(\xi)-g(x)&f(x)\end{pmatrix}
$$
where $f(x)=\frac12(U^+(x)+U^-(x))=1-\cos(2\pi x)$ and $g(x)=\frac12(U^+(x)-U^-(x))=\sqrt 3\sin(2\pi x)$, and
$$
H_1(x,\xi)\equiv k\begin{pmatrix} 0&-i\\ i&0 \end{pmatrix},
$$
where $k=-k_x$ for the low-energy model, and $k=0$ for the tight-binding model. 
Using the Weyl calculus it is straightforward to check that
$$
D_cD_c^*=(H_0^w+hH_1^w)^2,
$$
so the non-zero eigenvalues of $H_{\mathrm{TM}}$ can be described via a Bohr--Sommerfeld rule for $H_0^w+hH_1^w$.

The eigenvalues of the principal symbol $H_0$ are
$$
\lambda_\pm(x,\xi)=f(x)\pm\sqrt{\Upsilon_0(\xi)^2+g(x)^2}.
$$
Let's consider the energy level $E_0=0$. Due to periodicity of  $H_0$, the pre-image of $\det(H_0(x,\xi))=0$ consists of a discrete infinite set of points where both eigenvalues vanish (this is the zero set of $\lambda_+$), and an infinite set of simple closed curves where $\lambda_-=0$, $\lambda_+>0$ and $d\lambda_-\ne0$, such that $\lambda_-$ has a barrier in each domain enclosed by one of the curves. (Such a curve thus corresponds thematically to an excited state for \eqref{eq:simpleDirac} near an energy like $E_0=-1$.) The spectrum coming from the discrete set of points where $\lambda_+=\lambda_-=0$ was analyzed in \cite{becker2022semiclassical}.

Take a fundamental domain $\Omega$ of $\lambda_-$ containing precisely one such simple closed curve. We can for example take $\Omega=[0,1)\times\R$ in the low-energy case, and $\Omega=[0,1)\times[0,1)$ in the tight-binding case. Let $\gamma\subset\Omega\cap\lambda_-^{-1}(0)$ be the corresponding simple closed curve, see Figure \ref{fig:detH0}. Since $\lambda_-$ has a barrier in the domain enclosed by $\gamma$ it follows that $\gamma$ is oriented in the counterclockwise direction.
In view of Remark \ref{rem:manyconnectedcomponents} and Theorem \ref{thm:barrier} the spectrum of $H_0^w+hH_1^w$ coming from $\gamma$ is described by the Bohr--Sommerfeld rule $$2\pi kh=-S_0(E)+h(\pi-\pi\operatorname{wind}(\Gamma,0))+\mathcal O(h^2),\quad k\in\Z,$$ where $S_0(E)=\int_\gamma \xi\,dx$ and $\Gamma=q(\gamma)$ with $q(x,\xi)=-g(x)-i\Upsilon(\xi)$.

\begin{figure}
\includegraphics[width=0.45\textwidth]{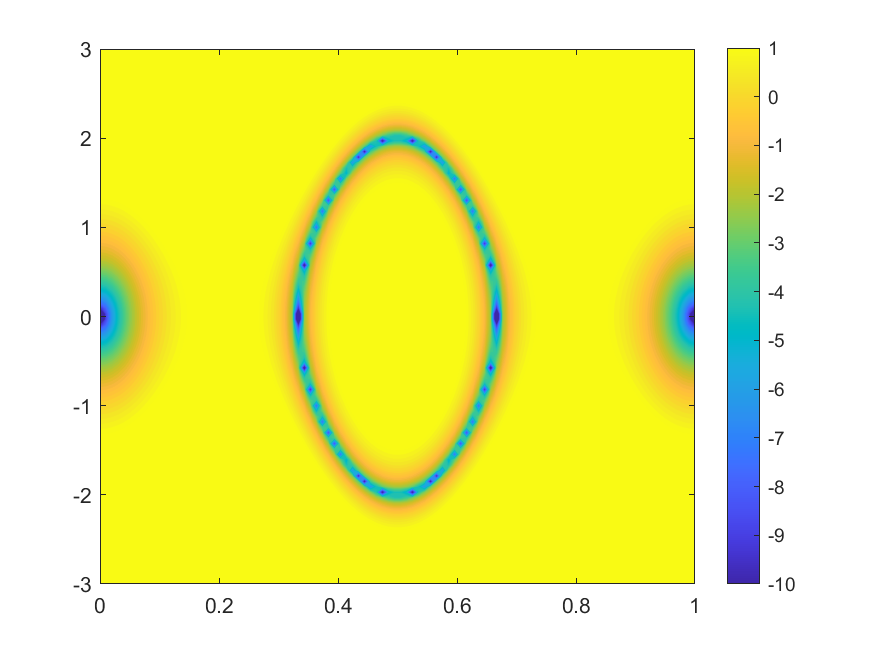}
\includegraphics[width=0.45\textwidth]{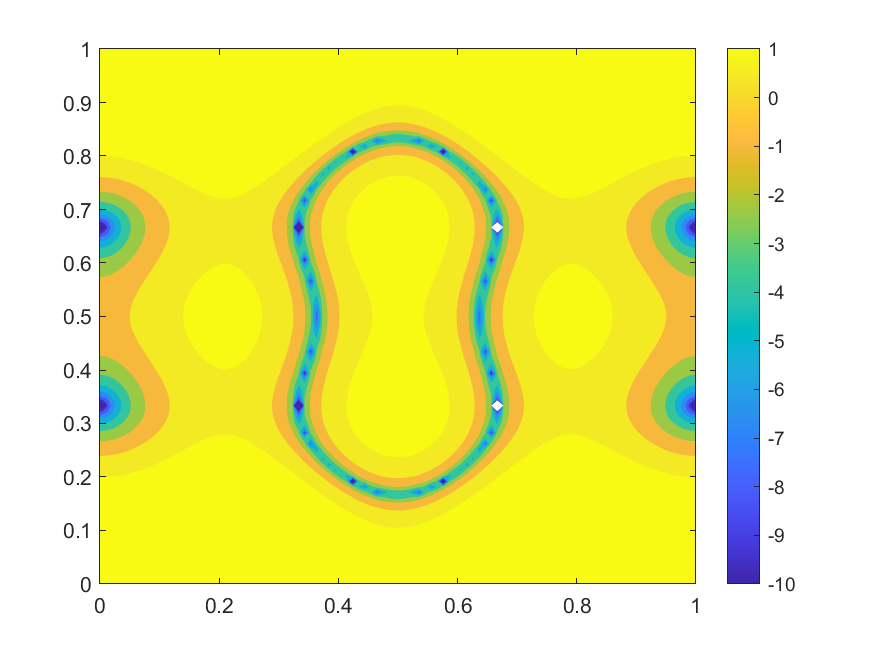}
\caption{Contour plot of the logarithm of the determinant of $H_0(x,\xi)$ in a fundamental domain, showing the zero set consisting of the discrete set where $\lambda_+=\lambda_-=0$ together with the simple closed curve $\gamma$ for the low-energy model (left) and tight-binding model (right).\label{fig:detH0}}
\end{figure}

For the low-energy model we have $\Upsilon_0(\xi)=\xi$. 
In the domain enclosed by $\gamma$, $\lambda_-(x,\xi)$ has a conic singularity at $(x,\xi)=(\frac12,0)$, and $\re q(x,\xi)=-g(x)<0$ when $0<x<\frac12$, while $\re q(x,\xi)=-g(x)>0$ when $\frac12<x<1$. Of course $\Upsilon_0(\xi)>0$ when $\xi>0$ and $\Upsilon_0(\xi)<0$ when $\xi<0$. It is then easy to see that as $\gamma$ winds once counterclockwise around $(\frac12,0)$, $\Gamma$ winds once clockwise around the origin in $\C$, so $\operatorname{wind}(\Gamma,0)=-1$. This yields a Bohr--Sommerfeld rule 
\begin{equation}
\label{eq:low-energy}
S_0(E)=2\pi k h + \mathcal O(h^2).
\end{equation}

For the tight-binding model we have $\Upsilon_0(\xi)=2\cos(2\pi\xi)+1$. In the domain enclosed by $\gamma$, $\lambda_-(x,\xi)$ has two conic singularities: at $(x,\xi)=(\frac12,\frac13)$ and $(x,\xi)=(\frac12,\frac23)$. As before $\re q(x,\xi)=-g(x)<0$ when $0<x<\frac12$, while $\re q(x,\xi)=-g(x)>0$ when $\frac12<x<1$, but now $\Upsilon_0(\xi)>0$ when $\xi>\frac23$ or $\xi<\frac13$, while $\Upsilon_0(\xi)<0$ when $\frac13<\xi<\frac23$.  It is then easy to see that as $\gamma$ winds once around $(\frac12,0)$, $\Gamma$ completes two major but incomplete arcs of opposite direction, never intersecting the negative real half-line $\re q<0$. (This is because the curve in the right panel of Figure \ref{fig:detH0} is not pinched all the way to $x=\frac12$ when $\frac13<\xi<\frac23$.) Hence, $\Gamma$ does not complete a full circuit around the origin in $\C$, so $\operatorname{wind}(\Gamma,0)=0$. This yields a Bohr--Sommerfeld rule 
\[ S_0(E)= (2k+1) \pi h + \mathcal O(h^2).\]

\subsubsection{Flat bands} For the low-energy model we find from the Weyl calculus that there will be a contribution to higher order terms $S_j(E)$ coming from 
$$
\Big(\frac{ih}{2}\Big)^k(D_\xi D_y-D_xD_\eta)^k \bar e(x,\xi) \cdot (H_1 e(y,\eta)).
$$
When $e=u_\pm$ we get 
$$
\pm ik_x\Big(\frac{ih}{2}\Big)^k(D_\xi D_y-D_xD_\eta)^k (e^{i\phi(y,\eta)}-e^{-i\phi(x,\xi)})
$$
which vanishes when  $k\ge1$. Since the term with $k=0$ yields no contribution (Proposition \ref{prop:quantizedH1}) we see that the bands, i.e., the $k_x$-parametrized eigenvalues of $H_{\mathrm{TM}}$, are independent of the quasimomentum $k_x$ modulo $\mathcal O(h^\infty)$. This effect is clearly visible in Figure \ref{fig:h}. We note that the curve contributes more eigenvalues than the well near zero energy, which is explained by the $\sqrt h$ scaling of the energy levels coming from the respective Bohr--Sommerfeld rules. Without going into details, we mention that the notion of flat bands also exists for the tight-binding model; for more in this direction we refer the reader to \cite[Section 1.3]{becker2022semiclassical}.

\begin{figure}
\includegraphics[width=6.5cm]{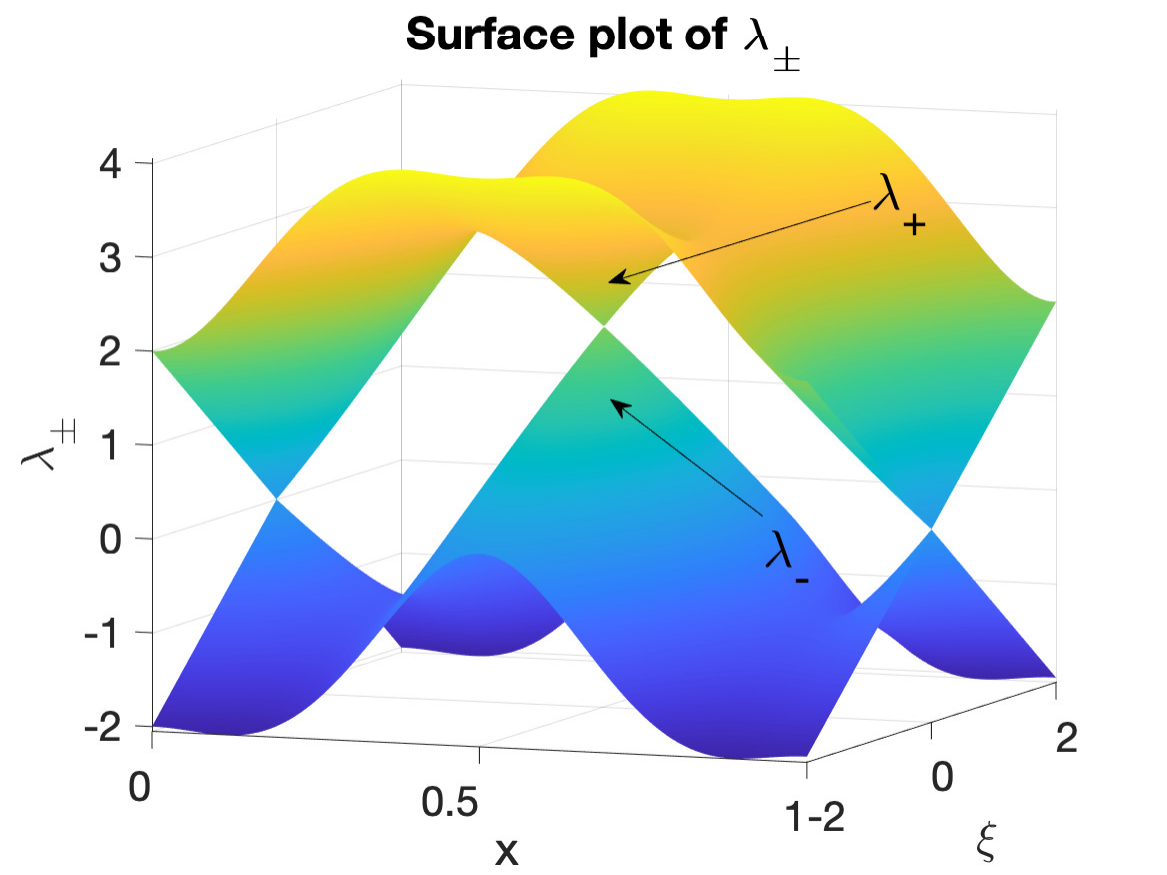}
\includegraphics[width=6.5cm]{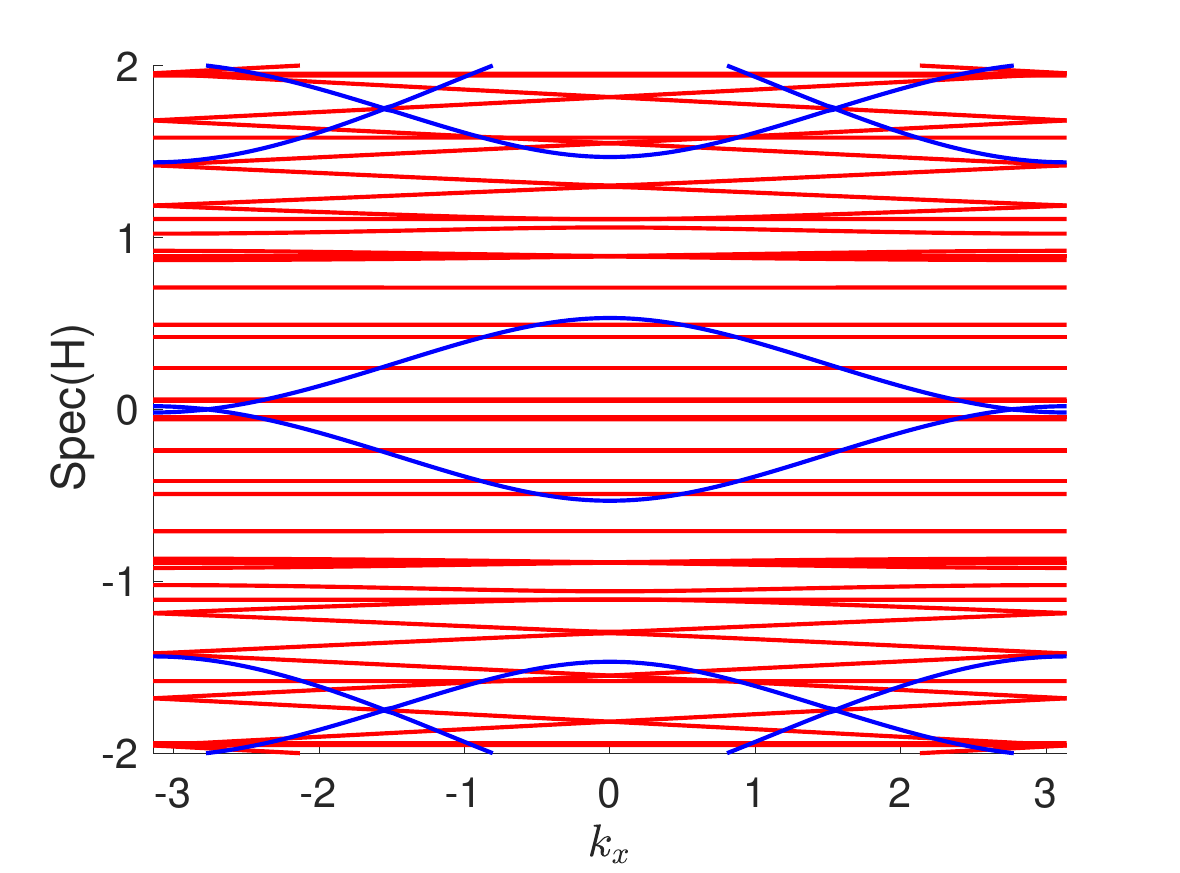}
\includegraphics[width=6.5cm]{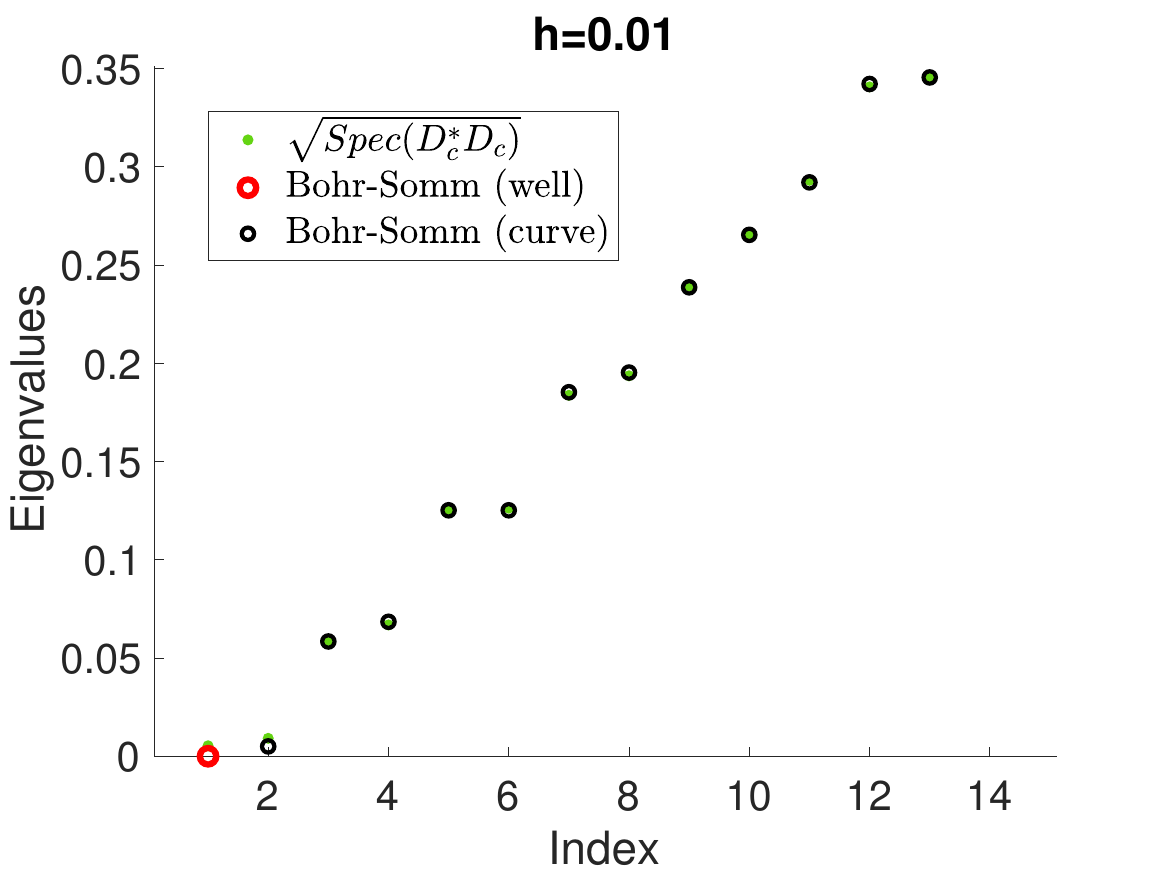}
\includegraphics[width=6.5cm]{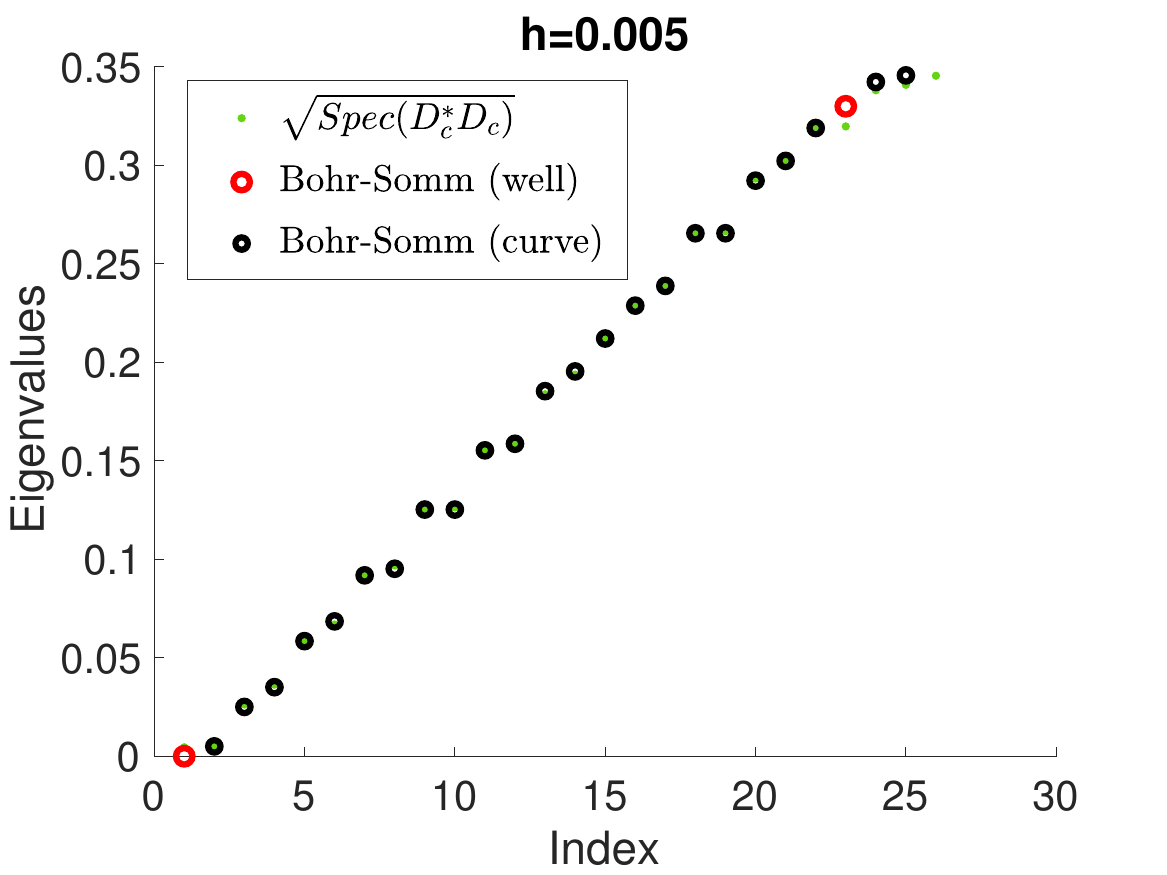}
\caption{\label{fig:h} For low-energy model with $\Upsilon(\xi)=\xi$:\\
(Top left): The eigenvalues $\lambda_{\pm}$ of $H_0.$ 
(Top right): The bands closest to zero of $H_{\mathrm{TM}}$ for {\color{blue}$h=0.5$} and {\color{red}$h=0.05$}.\\
(Bottom): Singular values of $D_c$, approximate eigenvalues from curve \eqref{eq:low-energy} and approximate eigenvalues from wells \cite[Theorem 1.6]{becker2022semiclassical} for $h=0.01$ (left) and $h=0.005$ (right).}
\end{figure}

\subsection{Radially symmetric Dirac operator}
In addition to the examples above, a natural source for the kind of one-dimensional self-adjoint operators studied in this work are radially symmetric Dirac operators in $\mathbb R^3$. They are described by a matrix-valued vector $\alpha = (\alpha_1,\alpha_2,\alpha_3)$ with matrix-valued entries $\alpha_i = \begin{pmatrix} 0 & \sigma_i \\ \sigma_i & 0 \end{pmatrix} \in \mathbb C^{4\times 4}$ and another $\beta = \operatorname{diag}(\sigma_0,-\sigma_0) \in \mathbb C^{4\times 4}$. The Dirac operator is a $4\times 4$ matrix-valued operator
\[\begin{split} H &= -i \alpha \cdot \nabla +\beta mc^2 + V(x) \text{ on } C_c^{\infty}(\mathbb R^3)
\end{split}\] 
with a potential given by
\[ V(x) =  \operatorname{id}_{\mathbb C^4} \phi_{0}(r) + i \beta \left(\alpha \cdot \frac{x}{\vert x \vert}\right)  \phi_{1}(r)+\beta \phi_{3}(r).\]
The potentials $\phi_i$ in $V$ have distinct physical interpretations. 
\begin{itemize}
 \item $\phi_0=\phi_{\mathrm{el}}$ plays the role of an \emph{electric potential} and does not lead to confinement.
  \item $\phi_1=\phi_{\mathrm{am}}$ is often described as the \emph{anomalous magnetic potential}.
\item $\phi_3=\phi_{\mathrm{sc}}$ is commonly referred to as the \emph{scalar potential} and having $\lim_{r\to \infty }\phi_3(r) = \infty$ leads to confinement of both electrons and positrons and therefore naturally leads to discrete spectrum. 
 \end{itemize}
 It is known \cite[Theorem 4.14]{thaller2013dirac} that $H$ is unitarily equivalent to the direct sum of \emph{partial wave} Dirac operators $h_{m_j,\kappa_j}$
\[ H \simeq \bigoplus_{j \in \frac{2 \mathbb N_0+1}{2}} \bigoplus_{m_j =-j}^j \bigoplus_{\kappa_j = \pm (j+1/2)} h_{m_j,\kappa_j}.\]
Particle wave Dirac operators are $2\times 2$ matrix-valued operators given for $V(r) = \sum_{i =0; i \neq 2}^4 \sigma_i \phi_i(r)$ by
\[ \begin{split} h_{m_j,\kappa_j} &= \sigma_2 D_r + V(r) + \sigma_1 \frac{\kappa_j}{r} \text{ on } C_c^{\infty}(0,\infty) \\
&= \sigma_2 D_r + V(r) + \sigma_1 \frac{\kappa_j}{r} \text{ on } C_c^{\infty}(0,\infty) \\
\end{split}\]
where more explicitly
\[ h_{m_j,\kappa_j} = \begin{pmatrix} mc^2+ \phi_0+ \phi_3 & \frac{\kappa_j}{r}+ \phi_1 -\partial_r \\  \frac{\kappa_j}{r}+ \phi_1 +\partial_r & -mc^2 + \phi_0 -\phi_3  \end{pmatrix}.\]
The operators $h_{m_j,\kappa_j}$ can, for suitable energies and suitable choices of the potential, be studied using the Bohr--Sommerfeld rule discussed in this work.

\section*{Acknowledgments}
We wish to thank Michael Hitrik and Bernard Helffer for interesting and helpful discussions. The research of Jens Wittsten was supported by The Swedish Research Council grant 2023-04872. The research of Setsuro Fujiié was supported by JSPS Kakenhi grant number 24K06790. Simon Becker acknowledges support from the SNF grant PZ00P2 216019.

\end{document}